\documentclass[review]{elsarticle}

\usepackage{color}
\usepackage[textsize=scriptsize]{todonotes}
\usepackage{graphics} 
\usepackage{epsfig} 
\usepackage{amsmath} 
\usepackage{amssymb}  

\usepackage{amsthm}
 \usepackage{tabularx}
\usepackage{dsfont}
\usepackage{bm} 
\usepackage{url}
\usepackage{balance}
\usepackage{color}
\usepackage{graphicx}
\usepackage{subfig}
\usepackage{mathtools}
\usepackage{arydshln}

\usepackage[font=small,labelfont=bf]{caption}
\graphicspath{{./Figures/}}
\usepackage{array}
\newcolumntype{L}[1]{>{\raggedright\let\newline\\\arraybackslash\hspace{0pt}}m{#1}}
\newcolumntype{C}[1]{>{\centering\let\newline\\\arraybackslash\hspace{0pt}}m{#1}}
\newcolumntype{R}[1]{>{\raggedleft\let\newline\\\arraybackslash\hspace{0pt}}m{#1}}

\newtheorem{theorem}{Theorem}
\newtheorem{definition}{Definition}
\newtheorem{assumption}{Assumption}
\newtheorem{lemma}{Lemma}
\newtheorem{corollary}{Corollary}
\newtheorem{remark}{Remark}
\newtheorem{proposition}{Proposition}

\newtheorem{problem}{Problem}

\usepackage[bookmarks=true]{hyperref}
\usepackage{algorithm,algorithmicx,algpseudocode}
\algnewcommand{\algorithmicgoto}{\textbf{go to}}%
\algnewcommand{\Goto}[1]{\algorithmicgoto~\ref{#1}}%
\algnewcommand{\LineComment}[1]{\Statex \(\triangleright\) #1}
\algnewcommand{\LineCommentN}[1]{\Statex \hspace{1cm}\(\triangleright\) #1}

\usepackage{multirow}
\usepackage{stfloats}

\DeclareFontFamily{OT1}{pzc}{}
\DeclareFontShape{OT1}{pzc}{m}{it}{<-> s * [1.10] pzcmi7t}{}
\DeclareMathAlphabet{\pzc}{OT1}{pzc}{m}{it}
\newcommand\Item[1][]{%
  \ifx\relax#1\relax  \item \else \item[#1] \fi
  \abovedisplayskip=0pt\abovedisplayshortskip=0pt~\vspace*{-\baselineskip}}

\usepackage{lineno,hyperref}
\modulolinenumbers[5]

\journal{Journal of \LaTeX\ Templates}

\newcommand{\yong}[1]{{\color{black} #1}}
\newcommand{\moh}[1]{{\color{black} #1}}

\begin{document}

\begin{frontmatter}

\title{Simultaneous Mode, State and Input Set-Valued Observers for \yong{Switched Nonlinear}  Systems} 



\author{Mohammad Khajenejad}
\ead{mkhajene@asu.edu}


\author{Sze Zheng Yong\corref{mycorrespondingauthor}}
\cortext[mycorrespondingauthor]{Corresponding author}
\ead{szyong@asu.edu}
\address{School for Engineering of Matter, Transport and Energy, Arizona State University, Tempe, AZ, USA}





\begin{abstract}
In this paper, we study the problem of designing \yong{a} simultaneous mode, input and state set-valued observer for a class of hidden mode switched nonlinear systems with bounded-norm noise and 
unknown input signals, \yong{where the hidden mode and unknown inputs can represent fault or attack models and exogenous \yong{fault/}disturbance or adversarial signals, respectively.} The \yong{proposed multiple-model} design 
has three constituents: (i) a bank of mode-matched \yong{set-valued} observers, 
 (ii) a mode 
 \yong{observer} and (iii) a global fusion observer.  
The mode-matched \yong{observers recursively find the} sets of compatible states and unknown inputs \yong{conditioned on the mode being the true mode,} 
while the mode 
\yong{observer} eliminates incompatible modes \yong{by leveraging a residual-based criterion}. Then, the global fusion observer outputs the estimated sets of states and unknown inputs by taking the union of the mode-matched \yong{set-valued} estimates over all compatible modes. Moreover, 
sufficient conditions to guarantee the elimination of all false modes (i.e., mode detectability) are provided and 
		the effectiveness of our approach is 
		\yong{demonstrated and compared with existing approaches} using an illustrative example. 
\end{abstract}

\begin{keyword}
Fault detection\sep Mode estimation\sep Set-valued observers\sep Switched systems\sep Nonlinear systems
\MSC[2010] 00-01\sep  99-00
\end{keyword}

\end{frontmatter}


\section{Introduction}
\yong{Cyber-Physical Systems (CPS), which tightly couple communication and computation elements, can enhance the functionality of control systems and improve their performance. However, these features may also become a source of vulnerability to attacks or faults. On the other hand, autonomous systems, e.g., self-driving cars or robots, typically must operate without the direct knowledge of the intentions and decisions of other systems/agents. These systems, which can be conveniently considered within the general framework of hidden mode hybrid/switched systems (HMHS, see, e.g., \cite{verma2011safety,yong2018switching,yong2021simultaneous} and references therein) with unknown inputs, are often safety-critical. Thus, the ability to estimate the states, unknown inputs and modes of such systems is important for monitoring these systems as well as for designing feedback controllers with safety and security guarantees.}

\emph{Literature review.} The problem of designing filters/observers for hidden mode systems, without considering unknown inputs/faults/data injection attacks, has been extensively studied, e.g., in \cite{Bar-Shalom.2002,Mazor.1998} and references therein. \yong{Recently,} the work in \cite{yong2018switching,yong2021simultaneous} proposed an extension to include unknown inputs for stochastic systems, aiming to obtain \emph{point} estimates, i.e., the most likely or best single estimates. \yong{However, probabilistic distributions of uncertainty are often unavailable and moreover, it may also be desirable to consider set-valued uncertainties, e.g., bounded-norm noise, especially when hard guarantees or bounds are important. In the latter setting, \emph{set-membership} or \emph{set-valued} state observers, e.g., \cite{dahleh1994control,shamma1999set,blanchini2012convex}, have been proposed to estimate the set of compatible states, and later, extensions of this framework  to include the estimation of unknown inputs have been proposed in \cite{yong2018simultaneous,khajenejad2019acc,khajenejad2020simultaneousnonlinear}. Nonetheless, these approaches are not directly applicable to systems with hidden modes that are considered in this paper.}

To consider hidden modes, \yong{which can be used to model/represent fault or attack models,} 
a common approach is to construct \emph{residual} signals (see, e.g., \cite{yong2018switching,yong2021simultaneous,Bar-Shalom.2002,patton2013issues,giraldo2018survey,Pasqualetti.2013}),  
where 
a threshold 
based on the residual signal is used to 
distinguish between 
consistent and inconsistent modes. 
\yong{In the context of resilient state estimation against sparse data injection attacks,}
\cite{nakahira2018attack} presented a robust control-inspired 
\yong{approach for linear systems}
{with bounded-norm noise that} 
consists of 
local estimators, 
residual detectors, and a global fusion detector. 
Similar residual-based techniques have been used for uniformly observable nonlinear systems 
in \cite{kim2018detection} \yong{and some classes of nonlinear systems in \cite{chong2020secure}. However, these approaches only consider sparse attacks on the sensors, which is a special case of a hidden mode system, as was discussed in our previous work for  hidden mode switched linear stochastic systems in \cite{yong2018switching}. Thus, to our best knowledge, the design of an estimator for hidden mode switched nonlinear systems with unknown inputs and bounded-norm noise remains an open problem.}

%

 \emph{Contributions.} 
 \yong{To bridge this gap, this paper considers the problem of simultaneous mode, state and unknown input estimation for hidden mode switched nonlinear systems with bounded-norm noise, where the hidden mode represents a fault or attack model. To tackle this problem, our preliminary conference publication \cite{khajenejad2019simultaneousmode} proposed a multiple-model approach for  
 switched linear systems. In this paper, we further extend this approach to 
 hidden mode switched nonlinear systems with unknown inputs using a similar multiple-model approach, which consists of}  
 a bank of mode-matched set-valued observers and a novel \emph{elimination-based} mode observer. \yong{The mode-matched set-valued observers are based on the optimally designed set-valued state and input $\mathcal{H}_\infty$ observers in our recent work \cite{khajenejad2020simultaneousnonlinear}, while the mode observer eliminates} 
%
%
 inconsistent modes from the bank of observers by using the upper bound of the norm of 
 to-be-designed residual signals \yong{as a threshold. In particular, we propose a tractable method to calculate an upper bound signal for the residual's norm by carefully over-approximating} the value function of a non-concave NP-hard norm-maximization problem \yong{with}  
%
 a convex maximization \yong{problem} over a convex set that has a finite number of extreme points \yong{in a manner that guarantees that no compatible modes are eliminated.} 
 We also prove that the upper bound signal is a convergent sequence. 
  Furthermore, we provide sufficient conditions for \emph{mode detectability}, i.e., for guaranteeing that all false modes will be eventually ruled out under some reasonable assumptions. Finally, we compare the performance of our proposed approach with \yong{an existing $\mathcal{H}_\infty$ observer in the literature.} 

\emph{Notation.} 
$\mathbb{R}^n$ denotes the $n$-dimensional Euclidean space, and 
 $\mathbb{N}$ 
 \yong{the set of} nonnegative integers. 
 For a vector $v \in \mathbb{R}^n$, $\|v\|_2 \triangleq \sqrt{v^\top v}$ and $\| v \|_{\infty} \triangleq \max \limits_{1 \leq i \leq n} v_i$, and for a matrix $M \in \mathbb{R}^{p \times q}$, $\|M\|_2$, $\sigma_{\min}(M)$ and $M(i:j)$ denote \yong{the induced $2$-norm, the smallest non-trivial 
 singular value and the sub-matrix consisting of the $i$-th through $j$-th columns of $M$, respectively. Further, $0_{n \times m}$ denotes an $n$-by-$m$ zero matrix. }
\section{Problem Statement}
Consider a 
 hidden mode switched nonlinear system with bounded-norm noise and unknown inputs (i.e., 
 a hybrid system with nonlinear and noisy system dynamics in each mode, where 
 the mode and some inputs are not known/measured): 
\begin{align} \label{eq:sys_desc}
\begin{array}{ll}
\hspace{-0.25cm}
 x_{k+1}&=f^q (x_k)+B^q u^q_k+G^{q} d^q_{k} +W^q w^q_k, 
\\ 
y_k&=C^q x_k+ D^q u^q_k +H^{q} d^q_k+ v^q_{k},
\end{array}
\end{align}
where $x_k \in \mathbb{R}^n$ is the continuous system state and $q \in \mathbb{Q}=\{1,2,\dots,Q\} \subset \mathbb{N}$ is the hidden discrete state or \emph{mode}. For each $q \in \mathbb{Q}$, $y_k \in \mathbb{R}^l$ is the measurement output signal and $w^q_k \in \mathbb{R}^n$ and $v^q_k \in \mathbb{R}^l$ are \yong{external} process and measurement disturbances with known \yong{$\ell_2$-norm bounds, i.e., $\|w_k\|_2\le \eta_w$ and $\|v_k\|_2 \le \eta_v$,}  
respectively. Moreover, $u^q_k \in U_{k} \subset \mathbb{R}^m$ is the \emph{known} input and $d^q_k \in \mathbb{R}^p$ the unknown input signal \yong{(representing, e.g., the input of other agents/robots or adversarially injected data signal)}. 
It is worth mentioning that no prior `useful' knowledge or assumption of the dynamics of $d^q_k$ is assumed. For each (fixed) mode $q$, the mapping $f^q(\cdot):\mathbb{R}^n \to \mathbb{R}^n$ and the matrices $B^q \in \mathbb{R}^{n \times m}$, $G^q \in \mathbb{R}^{n \times p}$, $C^q \in \mathbb{R}^{l \times n}$, $D^q \in \mathbb{R}^{l \times m}$ and $H^q \in \mathbb{R}^{l \times p}$ are 
 the corresponding mode-dependent known state vector field and system matrices, respectively. 
 
 \yong{The above modeling framework can capture a very broad range of problems, including intention estimation, fault detection and resilient state estimation against sparse data injection and switching/mode attacks. Specifically, in the context of intention estimation or fault diagnosis, each mode represents an intent or fault model and the unknown inputs can model the inputs of other agents/robots or exogenous fault signals. On the other hand, with regard to resilient state estimation, the switching/mode attacks (e.g., attacks on circuit breakers) can be represented with a set of different $f^q(\cdot)$, $B^q$, $C^q$ and $D^q$, while the unknown attack location of sparse data injection attacks can be modeled by a set of different $G^q$ and $H^q$ that represent the different hypotheses for which actuators and sensors are attacked or not attacked. Further, the attack signal magnitudes can be modeled as the unknown inputs in this scenario.} 

 In addition, we 
 assume the following: 
 \begin{assumption}\label{ass:true_mode}
 There is only one ``true" mode, i.e. the true mode $q^*$ is constant over time.
 \end{assumption}\label{ass:Lip}
  \begin{assumption}
For each $q \in \mathbb{Q}$, $f^q(\cdot)$ is twice continuously differentiable and Lipschitz continuous on its domain with 
\yong{a} known Lipschitz constant $L^q_f >0$.
 \end{assumption}

\yong{Using} the above modeling framework, the 
simultaneous state, unknown input and hidden mode estimation problem \yong{based on a multiple-model framework}  
can be stated as follows:
\begin{problem}
Given a \yong{hidden mode switched nonlinear discrete-time system with unknown inputs and bounded-norm noise} in the form of \eqref{eq:sys_desc}, 
\begin{enumerate}[(i)]
\item Design a bank of mode-matched observers, \yong{where each mode-matched observer,} 
\yong{conditioned} on the mode being true, 
optimally returns the set\yong{-valued} estimates  
of compatible states and unknown inputs in the minimum $\mathcal{H}_\infty$-norm sense, i.e., with minimum average power amplification.
\item \yong{Find a threshold} 
criterion to eliminate false modes and \yong{subsequently,} 
develop a mode 
\yong{observer} via elimination. 
\item Derive sufficient conditions for \yong{the} elimination of all false modes.
\end{enumerate}
\end{problem}
\section{Proposed Observer Design}
In this section, we propose a multiple-model approach for simultaneous mode, state and unknown input estimation for the system in \eqref{eq:sys_desc}, with the 
\yong{goal} of \yong{recursively} finding \yong{the sets of states $\hat{X}_k$, unknown inputs $\hat{D}_k$ and modes $\hat{\mathbb{Q}}_k$ that are compatible with observed outputs $y_k$.} 
\subsection{Overview of Multiple-Model Approach} \label{sec:prelim}
 The multiple-model design approach 
 \yong{consists of} three steps: (i) designing a bank of mode-matched set-valued observers, (ii) \yong{developing} 
 a mode observer for eliminating incompatible modes using \yong{a residual-based threshold}, 
 and (iii) \yong{devising} a global fusion observer that returns the desired set-valued mode, input and state estimates. 
\subsubsection{Mode-Matched Set-Valued Observer} \label{sec:ULISE}
First, based on the optimal fixed-order observer {design} in \cite{khajenejad2020simultaneousnonlinear}, we develop a bank of mode-matched observers, 
\yong{which} includes $Q \in \mathbb{N}$ simultaneous state and input $\mathcal{H}_\infty$ set-valued observers, which can be briefly summarized as follows. 
For each mode-matched observer corresponding to mode $q$, 
following the approach in \cite[Section 4]{khajenejad2020simultaneousnonlinear}, we consider set-valued fixed-order estimates \yong{in the form of $\ell_2$-norm balls:} 
\begin{align}
\hat{D}^{q}_{k-1}&=\{d_{k-1} \in \mathbb{R}^p: \|d_{k-1}-\hat{d}^{q}_{k-1}\|\yong{_2}\leq \delta^{d,q}_{k-1}\},\\
\hat{X}^{q}_k&=\{x_k \in \mathbb{R}^n: \|x_k-\hat{x}^{q}_{k|k}\|\yong{_2} \leq \delta^{x,q}_k\},
\end{align}
where their centroids $\hat{x}^q_{k|k}$ and $\hat{d}^q_{k-1}$ 
are obtained with 
the following three-step recursive observer that is optimal in $\mathcal{H}_{\infty}$-norm sense (cf. \cite[Section 4.2]{khajenejad2020simultaneousnonlinear} for more details): 

\noindent\emph{Unknown Input Estimation}:
\begin{align}\label{eq:UIE}
\begin{array}{rl}
\hat{d}^{q}_{1,k} &=M^{q}_{1} (z^{q}_{1,k}-C^{q}_{1} \hat{x}^{q}_{k|k}-D^{q}_{1} u^{q}_k),\\
\hat{d}^{q}_{2,k-1}&=M^{q}_{2} (z^{q}_{2,k}-C^{q}_{2} \hat{x}^{q}_{k|k-1}-D^{q}_{2} u^{q}_k),\\
\hat{d}^{q}_{k-1}&= V^{q}_{1} \hat{d}^{q}_{1,k-1} + V^{q}_{2} \hat{d}^{q}_{2,k-1}; \end{array}
\end{align}
\emph{Time Update}:
\begin{align}
\hspace{-0.3cm}\begin{array}{rl}
 \hat{x}^{q}_{k|k-1}\hspace{-0.1cm}&=f^q( \hat{x}^{q}_{k-1 | k-1}) + B^q u^{q}_{k-1} + G^{q}_{1} \hat{d}^{q}_{1,k-1}, \\
\hat{x}^{\star,q}_{k|k}&=\hat{x}^{q}_{k|k-1}+G^{q}_{2} \hat{d}^{q}_{2,k-1}; 
\end{array}\hspace{-0.3cm}
\end{align}
\emph{Measurement Update}:
\begin{align}
\hat{x}^{q}_{k|k}
&= \hat{x}^{\star,q}_{k|k} +\tilde{L}^{q}  (z^{q}_{2,k}-C^{q}_{2} \hat{x}^{\star,q}_{k|k}-D^{q}_{2} u^{q}_k),  \quad \label{eq:stateEst}
\end{align}
where $C^q_1$, $C^q_2$, $D^q_1$, $D^q_2$, $G^q_1$, $G^q_2$, $V^q_1$, $V^q_2$, $z^q_{1,k}$ and $z^q_{2,k}$ can be computed by applying a similarity transformation described in \yong{\ref{app:transformation}} 
and 
$\tilde{L}^q \in \mathbb{R}^{n \times (l-p_{H^q})}$, $M^q_{1} \in \mathbb{R}^{p_{H^q} \times p_{H^q}}$ and $M^q_{2} \in \mathbb{R}^{(p-p_{H^q}) \times (l-p_{H^q})}$ are observer gain matrices that are chosen \yong{via} 
the following \yong{P}roposition \ref{prop:filterbanks}. This proposition is a restatement of the 
results in \cite{khajenejad2020simultaneousnonlinear} \yong{that is} tailored 
to the  setting \yong{considered} in this paper, where the main idea is to minimize the ``volume'' of the set of compatible states and unknown inputs, quantified by the radii $\delta^{d,q}_{k-1}$ and $\delta_k^{x,q}$. 
\begin{proposition}\cite[Proposition 5.16, Lemma 5.1 \& Theorem 5.13]{khajenejad2020simultaneousnonlinear} \label{prop:filterbanks}
Consider system \eqref{eq:sys_desc} and a bank of $Q$ mode-matched observers in the form of \eqref{eq:UIE}--\eqref{eq:stateEst}. Suppose that $\forall q \in \mathbb{Q}\triangleq \{1,\dots,Q\}$, ${\rm rk}(C^q_2G^q_2)=p-p_{H^q}$ and $M^q_{1},M^q_{2}$ are chosen as $M^q_{1}= (\Sigma^q)^{-1}$ and $M^q_{2}=(C^q_{2} G^q_{2})^\dagger$, where $\Sigma^q$ is obtained by applying singular value decomposition on $H^q$ \yong{(cf. \ref{app:transformation} for more details)}. 
Then, the following statements hold: 
\begin{enumerate}[(a)]
\item Given mode $q \in \mathbb{Q}$, the following difference equation governs the state estimation error dynamics {(i.e., {the dynamics of} $\tilde{x}^q_{k|k} \triangleq x_k-\hat{x}^q_{k|k}$)}:
\begin{align}\label{eq:errors-dynamics}
\tilde{x}^q_{k+1|k+1}=(I-\tilde{L}^qC^q_2)\Phi^q (\Delta f_k^q-\Psi^q \tilde{x}^q_{k|k})+\mathcal{W}^q(\tilde{L}^q)\overline{w}^q_k,
\end{align}
where 
\begin{align*}
\Delta f^q_k &\triangleq f^q(x_k)-f^q(\hat{x}^q_k), \quad \Phi^q \triangleq I-G^q_2 M^q_2 C^q_2,\\
 \overline{w}^q_k &\triangleq \begin{bmatrix} (\frac{1}{\sqrt{2}})v^{q\top}_k & w^{q\top}_k & (\frac{1}{\sqrt{2}})v^{q\top}_{k+1}  \end{bmatrix}^\top,\\
  R^q &\triangleq \begin{bmatrix} -\sqrt{2}\Phi^q G^q_1M^q_1T^q_1 & -\Phi^q W^q & -\sqrt{2}G^q_2M^q_2T^q_2  \end{bmatrix},\\
   Q^q &\triangleq \begin{bmatrix} 0_{(l-p_{H^q}) \times l} & 0_{(l-p_{H^q}) \times n} & -\sqrt{2}T^q_2 \end{bmatrix}, \\
    \Psi^q &\triangleq G^q_1M^q_1C^q_1, 
     \quad \mathcal{W}^q(\tilde{L}^q) \triangleq (I-\tilde{L}^qC^q_2)R^q+\tilde{L}^qQ^q.
    \end{align*}

\item Solving the following mixed-integer SDP for each mode $q$:
\yong{
\begin{align*}
&(\rho^{\star}_q)^2=\hspace{-.35cm}\min_{\{P \succ 0,\Gamma \succ 0,\tilde{\Gamma} \succeq 0,\breve{Q} \succeq 0,Y, \breve{Z},\rho^2>0,0 \leq\alpha \leq 1,\varepsilon_1>0,\varepsilon_2>0,{\kappa>0},\kappa_1 >0,\kappa_2>0\}} \rho^2 \\ 
& s.t. \begin{bmatrix} P & \tilde{Y}^q_1 \\ \tilde{Y}^{q\top}_1 & \tilde{\mathbf{M}}^q_{1}  \end{bmatrix} \succeq 0,\begin{bmatrix} P & \tilde{Y}^q_2 \\ \tilde{Y}^{q\top}_2 & \tilde{\mathbf{M}}^q_{2} \end{bmatrix} \succeq 0, \begin{bmatrix} P & \tilde{Y}^q_1 \\ \tilde{Y}^{q\top}_1 & \tilde{\mathbf{M}}^q_{3} \end{bmatrix} \succeq 0,\\
&\quad \ \begin{bmatrix} P & \tilde{Y}^q_2 \\ \tilde{Y}^{q\top}_2 & \ \breve{Z} \end{bmatrix} \succeq 0, \  \ \begin{bmatrix} \tilde{\Gamma} & \breve{Z} \\ \breve{Z}^\top  & \Psi^{q\top} \breve{Q} \Psi^q \end{bmatrix} \succeq 0, \\
&\quad \   \begin{bmatrix} I-\Gamma & 0 & 0 \\ 0 & P & Y \\ 0 & Y^\top & I \end{bmatrix} \succeq 0, \
 \begin{bmatrix} \mathcal{N}^q_{11} & * & *  \\ \mathcal{N}^q_{21} & \mathcal{N}^q_{22} & *  \\ \mathcal{N}^q_{31} & 0 & \mathcal{N}^q_{33}  \end{bmatrix} \succeq 0,\\
& \quad \ \ \kappa_1I \preceq P \preceq \kappa_2 I,  \   \land \ (  (\kappa_1 \geq 1, \kappa_2-\kappa_1 <1) \ \vee \ (\kappa_2 \leq 1, \kappa_1 >0.5)),
\end{align*}}
\yong{we obtain an observer in the form of \eqref{eq:UIE}--\eqref{eq:stateEst} with the observer gain $\tilde{L}^{q}=(P^q)^{-1}Y^{q}$, where $(P^{q},Y^{q})$ are solutions to the above mixed-integer SDP, that}
\begin{itemize}
 \item \yong{is quadratically stable,} 
 and 
\item guarantees that 
\begin{align}\label{eq:thetaq}
\theta^q \triangleq \|(I-\tilde{L}^qC^q_2)\Phi^q \|\yong{_2} <1,
\end{align}
and consequently, \yong{the upper bound sequences for the radii 
$\{{\delta}^{x,q}_k,\moh{{\delta}^{d,q}_{k-1}}\}_{k=1}^{\infty}$, 
which are computed as:} 
\begin{align}\label{eq:deltax}
\begin{array}{ll}
{\delta}^{x,q}_k &\triangleq \delta^x_0 \yong{(\theta^q)}^k+\overline{\eta}^q\frac{1-{\yong{(\theta}^q)}^k}{1-\theta^q}, \\ 
{\delta}^d_{k-1} &\triangleq \beta^q {\delta}^{x,q}_{k-1}+ \overline{\alpha}^q,
\end{array}
\end{align}
 are convergent to some steady state value \yong{${\delta}^{x,q}_{\infty},\moh{{\delta}^{d,q}_{\infty}}$} 
\yong{(cf. \ref{app:matrices} for definitions of ${\delta}^{x,q}_{\infty}$ and $\moh{{\delta}^{d,q}_{\infty}}$, as well as the matrices and parameters in the above SDP and \eqref{eq:deltax})}. 
\end{itemize}
\end{enumerate}
\end{proposition}
\subsubsection{Mode Observer}
To estimate the set of compatible modes, we consider an elimination approach that compares the $\ell_2$-norm of \emph{residual} signals against some thresholds. Specifically, we will eliminate a specific mode $q$, if $\|r^{q}_k\|_2>\hat{\delta}^{q}_{r,k}$, where the residual signal $r^{q}_k$ is defined as follows and the thresholds $\hat{\delta}^{q}_{r,k}$ will be derived in Section  \ref{sec:MainResult}.
\vspace{-0.1cm}
\begin{definition} [Residuals] \label{defn:computedresidual}
For each mode $q$ at time step $k$, 
 the residual signal is defined as:
\begin{align}
\nonumber r^{q}_k \triangleq z^{q}_{2,k}-C^{q}_{2} \hat{x}^{\star,q}_{k|k}-D^{q}_{2} u^{q}_k.
\end{align}
\end{definition}
\subsubsection{Global Fusion Observer}
Finally, combining the outputs of both components above, our proposed global fusion observer will provide
 mode, unknown input and state set-valued estimates at each time step $k$ as:
\begin{align*}
\begin{array}{c}
\hat{\mathbb{Q}}_k=\{q \in \mathbb{Q} \ \vline \ \|r^q_k\|_2 \leq \hat{\delta}^q_{r,k} \},
\\ \hat{D}_{k-1}=\cup_{q \in \hat{\mathbb{Q}}_k} D^q_{k-1},\  \hat{X}_k=\cup_{q \in \hat{\mathbb{Q}}_k} X^q_{k}.
\end{array}
\end{align*} 

\noindent The 
multiple-model approach is summarized 
 in Algorithm \ref{algorithm1}.

 \begin{algorithm}[!t] \small
\caption{Simultaneous Mode, State and Input Estimation }\label{algorithm1}
\begin{algorithmic}[1]
  \State $\hat{\mathbb{Q}}_0=\mathbb{Q}$;
  \For {$k =1$ to $N$}
  \For {$q \in \hat{\mathbb{Q}}_{k-1}$}
  \LineComment{Mode-Matched State and Input Set-Valued Estimates}
 \Statex \hspace{0.4cm} Compute $T^q_2,M^q_{1},M^q_{2},\tilde{L}^q,\hat{x}^{\star,q}_{k|k},\hat{X}^{q}_{k},\hat{D}^q_{k-1}$ via Proposition \ref{prop:filterbanks};
\Statex \hspace{0.4cm} $z^q_{2,k}=T^q_2y_k$;
 \LineComment{Mode Observer via Elimination}
 \Statex \hspace{0.4cm} $\hat{\mathbb{Q}}_k=\hat{\mathbb{Q}}_{k-1}$;
 \Statex \hspace{0.4cm} Compute $r^q_k$ via Definition \ref{defn:computedresidual} and $\hat{\delta}^q_{r,k}$ via Theorem \ref{thm:resid_comp_up_bound};
  \If {$\|r^q_k\|_2>\hat{\delta}^q_{r,k}$} $\hat{\mathbb{Q}}_k=\hat{\mathbb{Q}}_{k} \backslash \{q\}$;
 \EndIf
 \EndFor
 \LineComment{State and Input Estimates}
 \State $\hat{X}_k=\cup_{q \in \hat{\mathbb{Q}}_k} \hat{X}^q_k$; \ $\hat{D}_k=\cup_{q \in \hat{\mathbb{Q}}_k} \hat{D}^q_k$;
 \EndFor
\end{algorithmic}
\end{algorithm}
\vspace{-0.1cm}
\subsection{Mode Elimination Approach} \label{sec:MainResult}
\vspace{-0.05cm}
We leverage a relatively simple idea to develop a criterion for elimination of false modes, as follows. We rule out a particular mode as incompatible, if the $\ell_2$-norm of  
its corresponding residual signal exceeds its upper bound conditioned on this mode being true. 
To do so, for each mode $q$, we first compute an upper bound ($\hat{\delta}^q_{r,k}$) for the $\ell_2$-norm of its corresponding residual at time $k$, conditioned on $q$ being the \emph{true} mode. Then, comparing the $\ell_2$-norm of residual signal in Definition \ref{defn:computedresidual} with $\hat{\delta}^q_{r,k}$, mode $q$ can be eliminated if the residual's $\ell_2$-norm is strictly greater than the upper bound. 
 The following proposition and theorem formalize this procedure. 
\begin{proposition} \label{prop:residecomposition}
Consider mode $q$ at time step $k$, its residual signal $r^q_k$ (as defined in Definition \ref{defn:computedresidual}) and the unknown true mode $q^{*}$. Then, 
\begin{align}
\nonumber &r^q_k=r^{q|*}_k+\Delta r^{q|q*}_k, 
\end{align} 
{with}
\begin{align}
\nonumber &r^{q|*}_k \triangleq z^{q*}_{2,k}-C^q_2 \hat{x}^{\star,q}_{k|k}-D^q_2 u^q_{k}=T^{q*}_2y_k-C^q_2\hat{x}^{\star,q}_{k|k}-D^q_2 u^q_{k},
\\ \nonumber &\Delta r^{q|q*}_k \triangleq (T^q_2-T^{q*}_2)y_k,
 \end{align}
 {\!\!\!\! where $r^{q|*}_k$ 
is the true mode's residual signal (i.e., $q=q^*$), 
and $\Delta r^{q|q^*}_k$ 
is the \emph{residual error}.}

 \end{proposition}
  \vspace{-0.2cm}
 \begin{proof}
This follows directly from plugging the above expressions into the right hand side term of Definition \ref{defn:computedresidual}.
 \end{proof}
 \vspace{-0.3cm}
 \begin{theorem} \label{thm:online_mode_elimination}
 Consider mode $q$ and its residual signal $r^q_k$ at time step $k$.
 Assume that $\delta^{q,*}_{r,k}$ is any signal that satisfies $\| r^{q|*}_k\|_2 \leq \delta^{q,*}_{r,k}$, 
  where $r^{q|*}_k$ is defined in Proposition \ref{prop:residecomposition}. Then, mode $q$ is not the true mode, i.e., can be eliminated at time $k$, if $\| r^q_k \|_2 > \delta^{q,*}_{r,k}.$
 \end{theorem}
 \vspace{-0.2cm}
 \begin{proof}
 To use contradiction, suppose that $\| r^q_k \|_2 > \delta^{q,*}_{r,k}$ and let $q$ be the true mode, i.e., $q=q^*$ and thus, $T^q_2=T^{q*}_2$. 
 By Proposition \ref{prop:residecomposition}, $\Delta r^{q|q*}_k=0$ and hence, $\|r^q_k \|_2=\|r^{q|*}_k \|_2 \leq \delta^{q,*}_{r,k}$, which contradicts with the assumption. 
 \end{proof}
\yong{By the above theorem, our approach guarantees that the true mode is never eliminated. However, Theorem \ref{thm:online_mode_elimination} only provides 
a sufficient condition for mode elimination at each time step and the capability of our proposed mode observer to eliminate as many false modes as possible is dependent on the tightness of the upper bound, $\delta^{q,*}_{r,k}$.}
 
  \vspace{-0.1cm}
 \subsection{Tractable Computation of Thresholds} 
\vspace{-0.05cm}
To apply the sufficient condition \yong{in Theorem \ref{thm:online_mode_elimination}}, we need \yong{a tractable approach} to compute \yong{the upper bound $\delta^{q,*}_{r,k}$ that is finite-valued. This procedure is derived and described in the following.} 
 \vspace{-0.2cm}
\begin{lemma} \label{lem:resdef}
Consider any mode $q$ with the unknown true mode being $q^{*}$. Then, at time step $k$, we have  
\begin{align}
 r^{q|*}_k &= C^q_2 \tilde{x}^{\star,q}_{k|k}+v^q_{2,k}=\mathbb{A}^q_k {t}_k, \label{eq:resid_ideal}
 \end{align}
 where
 \begin{align*}
  &{t}_k \triangleq \hspace{-.1cm}\begin{bmatrix} \tilde{x}^{\top}_{0|0} & \hspace{-.2cm} v^{q\top}_0 & \dots & v^{q\top}_{k} & \hspace{-.2cm}w^{q\top}_0 & \dots & w^{q\top}_{k-1} &\hspace{-.2cm} \Delta f_0^{q\top} \dots \Delta f_{k-1}^{q\top} \end{bmatrix}^{\top}\hspace{-.2cm} \in \mathbb{R}^{{(n+l)(k+1)+nk}},\\
  &\mathbb{A}^q_k \triangleq [ A^q_k \ \ J^{q,1}_{k-1} \ \ (J^{q,2}_{k-1}+J^{q,1}_{k-2})  \cdots  (J^{q,2}_1+ J^{q,1}_0) \ \ J^{q,2}_0 \ \ J^{q,3}_{k-1} \dots J^{q,3}_{0} \ \ F^q_{k-1} \dots F^q_0 ],\\
  &A^q_k \triangleq (-1)^k((I-\tilde{L}^qC^q_2)\Phi^q\Psi^q)^k, \\
  & J^{q}_i \triangleq \begin{cases}\mathcal{Y}_q, \hfill \text{if} \ i=0, \\
  -C^q_2\Phi^qG^q_1M^q_1C^q_1(I-\tilde{L}^qC^q_2)^{i-1}\mathcal{W}^q, \quad \quad \quad \quad \quad \quad \quad \quad\quad \text{if} \ 1\leq i \leq k-1,
     \end{cases} \\
     & F^{q}_i \triangleq \begin{cases} C^q_2\Phi^q,\quad \quad \quad \quad \quad \quad \quad \quad\quad \quad \quad \quad \quad \quad \quad \quad \quad\quad \quad \quad \quad \quad \quad \quad \     \text{if} \ i=0, \\
  (-1)^{i}C^q_2\Phi^qG^q_1M^q_1C^q_1((I-\tilde{L}^qC^q_2)\Psi^q)^{i-1}(I-\tilde{L}^qC^q_2)\Phi^q, \ \ \text{if} \ 1\leq i \leq k-1,
     \end{cases} \\
     & \mathcal{Y}_q \triangleq \begin{bmatrix} -\sqrt{2}C^q_2\Phi^q G^q_1M^q_1T^q_1 & C^q_2\Phi^q W^q & \sqrt{2}(I-C^q_2G^q_2M^q_2)T^q_2  \end{bmatrix},\\
     & J^{q,1}_i \triangleq J^{q}_{i}(1:l), \ J^{q,2}_i \triangleq J^{q}_{i}(l+1:2l), \ J^{q,3}_i \triangleq J^{q}_{i}(2l+1:2l+n),  i=1,\dots, k-1.  
  \end{align*} 
\end{lemma}
\vspace{-0.2cm}\begin{proof}
The first equality \yong{in  \eqref{eq:resid_ideal}} comes from Definition \ref{defn:computedresidual} and $z^{q}_{2,k} = C^{q}_{2} x_k + D^{q}_{2,k} u^q_k + v^{q}_{2,k}$ from \yong{\eqref{eq:sysY} in \ref{app:transformation},} 
assuming that $q$ is the true mode. \yong{To obtain} 
the second equality, note that  \cite[(A.11)]{khajenejad2020simultaneousnonlinear} returns 
\begin{align}
\tilde{x}^{\star,q}_{k|k} &= \Phi^q [ \Delta f^q_{k-1} - G^q_{1}M^q_1C^q_1 \tilde{x}^q_{k-1|k-1}] \label{eq:xstar}
  +\textstyle  \tilde{w}^q_{k},
  \\ \nonumber  \tilde{w}^q_{k} &\triangleq -\Phi^q (G^q_1M^q_1v^q_{1,k-1}-W^q w^q_{k-1})-G^q_2M^q_2v^q_{2,k}.  
  \end{align}
  Now, from the first equality and \eqref{eq:xstar}, we have
  \begin{align} \label{eq:rqstar}
  r^{q|*}_k=C^q_2\Phi^q(\Delta f^q_{k-1}-G^q_1M^q_1C^q_1\tilde{x}^q_{k-1|k-1})+\mathcal{Y}^q \overline{w}^q_{k-1}.  
  \end{align}
  On the other hand, by iteratively applying \eqref{eq:errors-dynamics}, we obtain:
  \begin{align} \label{eq:tildexq}
 \nonumber \tilde{x}^q_{k|k}&=\sum_{i=1}^{i-1}[((I-\tilde{L}^qC^q_2)\Psi^q)^{i-1}(I-\tilde{L}^qC^q_2)\Phi^q\Delta f^q_{k-i}+(I-\tilde{L}^qC^q_2)^{i-1}\mathcal{W}^q\overline{w}^q_{k-i+1}]\\
  &+(-1)^k((I-\tilde{L}^qC^q_2)\Phi^q\Psi^q)^k\tilde{x}^q_{0|0}.
  \end{align}
  Combining \eqref{eq:rqstar} and \eqref{eq:tildexq} yields
  \begin{align*}
  r^{q|*}_k=A^q_k \tilde{x}^q_{0|0}+\sum_{i=0}^{k-1} F^q_i \Delta f^q_{k-1-i}+J^q_i \overline{w}^q_{k-i},
  \end{align*} 
  \yong{which} is equivalent to the second equality in \eqref{eq:resid_ideal}.
     \end{proof} 
     
\begin{lemma} \label{lem:existance}
For each mode $q$ at time step $k$, there exists a 
finite-valued upper bound $\delta^{q}_{r,k} < \infty$ for  $ \| r^{q|*}_k \|_2$.
 \end{lemma}
  \vspace{-0.2cm}\begin{proof}
 Consider the following optimization problem for {$ \| r^{q|*}_k \|_2$} by leveraging Lemma \ref{lem:resdef}:
\vspace{-0.05cm} \begin{align} \label{eq:genericupperbopund}
  &\delta^{q}_{r,k} \triangleq \max \limits_{t_k} \| \mathbb{A}^q_k {t}_k \|_2 
 \\ \nonumber  &s.t. \ t_k = \begin{bmatrix} \tilde{x}^{\top}_{0|0} &  v^{q\top}_0 & \dots & v^{q\top}_{k} & w^{q\top}_0 & \dots & w^{q\top}_{k-1} & \Delta f_0^{q\top} \dots \Delta f_{k-1}^{q\top} \end{bmatrix}^{\top},  
 \\ \nonumber &\| \tilde{x}_{0|0} \|_2 \leq \delta^x_0, \ \|v^q_i\|_2 \leq \eta^q_v, \ \| w^q_j \|_2 \leq \eta^q_w, \ \|\Delta f^q_j\|_2 \leq L^q_f \overline{\delta}^{x,q}_{j} \leq  L^q_f\overline{\delta}^{x,q}, 
 \\ \nonumber &i \in \{ 0,...,k \}, \ j \in \{ 0,...,k-1 \}.
 \end{align}
 The objective $\ell_2$-norm function is continuous and the constraint set is an intersection of level sets of lower dimensional norm functions, which is closed and bounded{,} so is compact. Hence, by the Weierstrass Theorem \cite[Proposition 2.1.1]{bertsekas2003convex}, the objective function attains its maxima on the constraint set and so a finite-valued upper bound exists. 
 \end{proof}
 
%


\vspace{-0.2cm}

Clearly, $\delta^{q}_{r,k}$ in Lemma \ref{lem:existance}, if computable, is the \emph{tightest} possible \yong{upper bound for the norm of the residual signal and using this as the threshold} 
can eliminate the most possible number of \yong{false} modes. 
However, note that although the existence \yong{proof} of a finite-valued $\delta^{q}_{r,k}$ is straightforward, the optimization problem in Lemma \ref{lem:existance} is NP-hard \cite{bodlaender1990computational}, since it is a \emph{norm maximization} (not minimization) over the intersection of level sets of lower dimensional norm functions, i.e., it is a non-concave maximization over intersection of quadratic constraints. To tackle 
this complexity, through the following Theorem \ref{thm:resid_comp_up_bound}, we propose a tractable over-approximation/upper bound for $\delta^{q}_{r,k}$, which we call $\hat{\delta}^q_{r,k}$ \yong{and is used instead as the elimination threshold}.
\begin{theorem} \label{thm:resid_comp_up_bound}
Consider mode $q$. At time step $k$, let
\begin{align}
\nonumber \hat{\delta}^q_{r,k} &\triangleq \min \{ {\delta}^{q,tri}_{r,k},\delta^{q,inf}_{r,k} \}, 
 \\   \delta^{q,tri}_{r,k} &\triangleq \sum_{i=0}^{k-2}L^q_f\|F^q_{i}\|_2\overline{\delta}^{x,q}_{k-1-i}+\frac{1}{\sqrt{2}}\eta^q_v(\|J^{q,1}_i\|_2+\|J^{q,3}_i\|_2)+\eta^q_w\|J^{q,2}_i\|_2 \label{eq:delta_tri} \\
\nonumber &+(\|A^q_k\|_2+L^q_f\|F^q_{k-1}\|_2)\delta^x_0+\frac{1}{\sqrt{2}}\eta^q_v(\|J^{q,1}_{k-1}\|_2+\|J^{q,3}_{k-1}\|_2)+\eta^q_w\|J^{q,2}_{k-1}\|_2,
 \\ \nonumber \delta^{q,inf}_{r,k} &\triangleq 
 \| \mathbb{A}^q_k {t}^{\star}_k \|_2, 
 \end{align}
 where \yong{$t_k^\star \triangleq \arg\max_{t_k \in \mathcal{T}_k} \| \mathbb{A}^q_k {t}_k \|_2$ and 
 $\mathcal{T}_k$ is the set of all vertices} of the following hypercube: 

\vspace{-0.3cm}
 \begin{align*}
 \begin{array}{l}
\mathcal{X}^q_k \triangleq \big\{ x \in \mathbb{R}^{{(n+l)(k+1)+nk}} \ \vline \\
| x(i) | \leq \begin{cases} \delta^x_0, \quad \quad   1 \leq i \leq n, \\
 \eta^q_v, \quad \quad   n+1 \leq i \leq n+{l(k+1)}, \\ 
 \eta^q_w, \quad \quad   n+{l(k+1)}+1 \leq i \leq {(n+l)(k+1)}, \\
L^q_f \delta^x_0, \quad   (n+l)(k+1)+1 \leq i \leq (n+l)(k+1)+n, \\
\vdots \\
L^q_f \overline{\delta}^{x,q}_j,   \ \ (n+l)(k+1)+jn+1 \leq i \leq (n+l)(k+1)+n(j+1), \\
\vdots \\
L^q_f \overline{\delta}^{x,q}_{k-1}, \  (n+l)(k+1)+(k-1)n+1 \leq i \leq (n+l)(k+1)+nk. \\
  \end{cases}\big\}.\end{array}
 \end{align*}
   Then, $\hat{\delta}^q_{r,k}$ is an over-approximation for $\delta^{q}_{r,k}$ in Lemma \ref{lem:existance}, \yong{i.e., $\hat{\delta}^q_{r,k} \ge {\delta}^q_{r,k}$}.
 \end{theorem}
  \vspace{-0.3cm}\begin{proof}
 Consider the following optimization problem:
  \begin{align}
  &\delta^{q,inf}_{r,k} \triangleq \max \limits_{t_k} \| \mathbb{A}^q_k {t}_k \|_2 \label{eq:inf_norm}
\\ \nonumber &s.t. \ t_k=t_k = \begin{bmatrix} \tilde{x}^{\top}_{0|0} &  v^{q\top}_0 & \dots & v^{q\top}_{k} & w^{q\top}_0 & \dots & w^{q\top}_{k-1} & \Delta f_0^{q\top} \dots \Delta f_{k-1}^{q\top} \end{bmatrix}^{\top},
 \\ \nonumber  &\ \ \ \ \  \| \tilde{x}_{0|0} \|_{\infty} \leq \delta^x_0, \ \|v^q_i\|_{\infty} \leq \eta^q_v, \ \| w^q_j \|_{\infty} \leq \eta^q_w, \ \|\Delta f^q_j\|_{\infty} \leq L^q_f \overline{\delta}^{x,q}_{j} 
 \\ \nonumber &\ \ \ \ \ \forall i \in \{ 0,...,k \}, \ \forall j \in \{ 0,...,k-1 \}. 
 \end{align}
 Comparing \eqref{eq:genericupperbopund} \yong{with} \eqref{eq:inf_norm}, the two problems have the same objective functions. \yong{Then,} 
 since $\| . \|_{\infty} \leq \| . \|_{2} $, the constraint set for \eqref{eq:genericupperbopund} is a subset of the one for \eqref{eq:inf_norm}. Hence $\delta^q_{r,k} \leq \delta^{q,inf}_{r,k}$. 
 Also, it is easy to see that 
$\yong{{\delta}^q_{r,k}}  \leq \delta^{q,tri}_{r,k}$, \yong{which is obtained using triangle inequality and the sub-multiplicative property of norms.} 
  Moreover, \eqref{eq:inf_norm} is a \emph{maximization} of a convex objective function over a convex constraint (hypercube $\mathcal{X}^q_k$). By a famous result \cite[Corollary 32.2.1]{rockafellar2015convex}, in such a problem, the objective function attains its maxima on some of the extreme points of the constraint set, which in this case are the vertices \yong{$\mathcal{T}_k$} of the hypercube $\mathcal{X}^q_k$. 
 \end{proof}

 Theorem \ref{thm:resid_comp_up_bound} enables us to obtain  an upper bound for $\|r^{q|*}_k\|_2$, by enumerating the objective function in \eqref{eq:inf_norm} \yong{for all} vertices of the hypercube $\mathcal{X}^q_k$ and choosing the largest value as $\delta^{q,inf}_{r,k}$. Moreover, we can easily calculate $\delta^{q,tri}_{r,k}$; then, the upper bound is chosen as 
 the minimum of the two as $\hat{\delta}^q_{r,k}$.
   \begin{remark}\label{rem:1} \yong{The reason for 
   not only using $\delta^{q,inf}_{r,k}$ is two-fold.}
 First, as time increases, the number of required enumerations \yong{for $\delta^{q,inf}_{r,k}$ (i.e., the cardinality of $\mathcal{T}_k$) can be shown to be $|\mathcal{T}_k|=2^{(n+l)(k+1)+kn}$, which increases at an exponential rate.} 
 Second and more importantly, as \yong{will be shown later in} Lemma \ref{lem:delta_inf_diverge}, 
 $\delta^{q,inf}_{r,k}$ goes to infinity as time increases, \yong{which renders it ineffective in the limit.} 
 \yong{On the other hand, Lemma \ref{lem:delta_inf_diverge} will show that $\delta^{q,tri}_{r,k}$ converges to some steady-state value, so it can always be used as an over-approximation for $\delta^{q}_{r,k}$ in the mode elimination process. Nonetheless, we chose to use the minimum of the two bounds, since our simulation results in Section \ref{sec:examples} show that $\delta^{q,inf}_{r,k}$ is generally smaller than $\delta^{q,tri}_{r,k}$ in the initial time steps.}
 \end{remark}
 
 \yong{Further, the following result that we will make use of later can be easily obtained as a corollary of Theorem \ref{thm:resid_comp_up_bound}.} 
 \begin{corollary} \label{cor:verticenorm}
 \yong{${t}^{\star}_k$ 
 (defined in Theorem \ref{thm:resid_comp_up_bound}) has the following norm:}
 \begin{align*}
 \eta^t_{k} \triangleq  \| {t}^{\star}_k \|_2=\sqrt{n( (1+{L^q_f}^2){\delta^x_0}^2+k {\eta^{q}_w}^2+{L^q_f}^2\sum_{j=1}^{k-1}{\overline{\delta}^{x,q}_j}^2)+l(k+1){\eta^{q}_v}^2}.
 \end{align*}
 \end{corollary}
\section{Mode Detectability}
 In addition to the nice properties regarding the quadratic stability and boundedness of the mode-matched set\yong{-valued} estimates of \yong{the} state and \yong{unknown} input obtained from  \cite{khajenejad2020simultaneousnonlinear}, \yong{we are interested in guaranteeing the effectiveness} 
 of our mode elimination algorithm. \yong{Thus, in the following, we search for} 
 some sufficient conditions 
 \yong{based on the properties/structures of} the system dynamics and/or unknown input signals \yong{for guaranteeing that}
 \yong{the application of Algorithm \ref{algorithm1} can eliminate} 
 \emph{all} false (i.e., not true) modes 
 after some large enough \yong{number of} time steps. 
 
 \yong{To achieve this, we first} define the concept of \emph{mode detectability}.  

\begin{definition}[Mode Detectability] \label{defn:strong_mode_detecatble}
System \eqref{eq:sys_desc} 
is called \emph{mode detectable} if 
 there exists a natural number $K>0$, such that for all time steps $k \geq K$, all false modes are eliminated.
\end{definition}

\yong{Moreover, we consider two different sets of assumptions that we will use} 
 for deriving our sufficient conditions for mode detectability.

%
\begin{assumption} \label{assumption:boundedness}
There exist known $R_y, R_x \in \mathbb{R}$ such that $\forall k, y_{k} \in Y \triangleq \{ y \in \mathbb{R}^{l} \vline \ \| y \|_2 \leq R_y \}$ and $x_{k} \in X \triangleq \{ x \in \mathbb{R}^{n} \vline \ \| x \|_2 \leq R_x \}$, i.e., there exist known bounds for the whole observation/measurement and state spaces, respectively. 
\end{assumption}


\begin{assumption}\label{as:2}
\yong{The state space $X$ is bounded and} the unknown input 
signal has  \emph{unlimited energy}, i.e., $\displaystyle\lim_{k\to\infty} \|d^{q*}_{0:k}\|_2 = \infty$, where $d^{q*}_{0:k} \triangleq \begin{bmatrix} d^{q*\top}_k & d^{q*\top}_{k-1} & \dots d^{q*\top}_0  \end{bmatrix}^{\top}$. 
\end{assumption}
Note that \yong{the unlimited energy condition in} Assumption \ref{as:2} is not restrictive \yong{if $f(\cdot)$, $B$, $C$ and $D$ are mode-independent, since} 
otherwise, the unknown input 
signal must vanish asymptotically, which means that \yong{we effectively have a non-switched system in the limit and the mode estimation would be trivial.} 

\yong{Next, in} order to derive the desired sufficient conditions for mode-detectability in Theorem \ref{thm:strong_mode_detect}, we first present the following  Lemmas \ref{lem:delta_inf_diverge}--\ref{lem:resdef2}. 
\begin{lemma} \label{lem:delta_inf_diverge}
 For each mode $q$,
 \begin{align}
 &\lim_{k\to\infty} \delta^{q,inf}_{r,k}= \infty. \label{eq:r_inf_diverge} \\
    &\lim_{k\to\infty} \hat{\delta}^{q}_{r,k}=\lim_{k\to\infty} \delta^{q,tri}_{r,k}< \infty ,  \label{eq:r_tri_converge}
     \end{align}
\begin{proof}
To show \eqref{eq:r_inf_diverge}, we first find a lower bound for $\delta^{q,inf}_{r,k}$. Then, we prove that the lower bound diverges and so does $\delta^{q,inf}_{r,k}$. Define $\tilde{t}^{\star}_k \triangleq \frac{t^{\star}_k}{\eta^t_k}$, where $\eta^t_k$ is defined in Corollary \ref{cor:verticenorm}. Now consider  
 \begin{align*}
 \eta^t_k \sigma_{min}(\mathbb{A}^q_k) =  \sigma_{min}(\eta^t_k \mathbb{A}^q_k)= \min \limits_{\| t \|_2 \leq 1} \| \eta^t_k \mathbb{A}^q_k t \|_2 
& \leq 
\| \eta^t_k \mathbb{A}^q_k \tilde{t}^{\star}_k\|_2 
=\| \mathbb{A}^q_k t^{\star}_k \|_2 \triangleq \delta^{q,inf}_{r,k},
 \end{align*}
 where $\sigma_{min}(A)$ is the 
 \yong{smallest} non-trivial singular value of matrix $A$. The first equality holds since $\sigma_{\min}(.)$ is a linear operator \yong{and} the second equality is a special case of \yong{the} \emph{matrix lower bound} \cite{grcar2010matrix} when $\ell_2$-norms are considered. The inequality holds since $\|\tilde{t}^{\star}_k\|_2=1$ by Corollary \ref{cor:verticenorm}, so $\tilde{t}^{\star}_k$ is a feasible point for the minimization \yong{problem (i.e., $\min \limits_{\| t \|_2 \leq 1} \| \eta^t_k \mathbb{A}^q_k t \|_2$)} 
 and the last equality holds by Theorem \ref{thm:resid_comp_up_bound}. So far we have shown that $\eta^t_k \sigma_{\min}(\mathbb{A}^q_k)$ is a lower bound for $\delta^{q,inf}_{r,k}$. {Next,} we will {prove} that $\eta^t_k \sigma_{min}(\mathbb{A}^q_k)$ is unbounded. First, it is trivial \yong{to observe} that $\eta^t_k$ 
 \yong{grows} unbounded by its definition {in} Corollary \ref{cor:verticenorm}. Second, $\sigma_{\min}(\mathbb{A}^q_k) \leq \sigma_{\min}(\mathbb{A}^q_{k+1})$, since the latter is an augmentation of the former with additional columns. 
 Hence, $\eta^t_k \sigma_{\min}(\mathbb{A}^q_k)$ \yong{grows}  unbounded, since the product of {the} unbounded and positive $\sigma_{\min}(\mathbb{A}^q_k)$ and the unbounded and positive $\eta^t_k$ is unbounded. 

To prove \eqref{eq:r_tri_converge}, \yong{we show that $\{\delta^{q,tri}_{r,k}\}_{k=1}^{\infty}$ is a convergent sequence. Then, this fact, as well as \eqref{eq:r_inf_diverge} and the fact that $\hat{\delta}^q_{r,k} \triangleq \min \{ {\delta}^{q,tri}_{r,k},\delta^{q,inf}_{r,k} \}$ by Theorem \ref{thm:resid_comp_up_bound}, imply \eqref{eq:r_tri_converge}. To show the convergence of $\{\delta^{q,tri}_{r,k}\}_{k=1}^{\infty}$, starting from \eqref{eq:delta_tri},} we first show that $\forall q \in \mathbb{Q}$, $S^q_{1,k} \triangleq \sum_{i=0}^{k-2}L^q_f\|F^q_{i}\|_2\overline{\delta}^{x,q}_{k-1-i}+\frac{1}{\sqrt{2}}\eta^q_v(\|J^{q,1}_i\|_2+\|J^{q,3}_i\|_2)+\eta^q_w\|J^{q,2}_i\|_2$ on the right hand side of \eqref{eq:delta_tri} converges to some steady state value. Note that $\|F^q_i\|_2 \leq \mathcal{R}^q {\theta^q}^i$ by \yong{the} sub-multiplicative \yong{property of norms,} 
where $$\mathcal{R}^q \triangleq L^q_f\|C^q_2\Phi^qG^q_1M^q_1C^q_1\|_2\|\Psi^q\|_2\|\Phi^q\|\moh{_2}$$ and $ {\theta^q}$ is given in \eqref{eq:thetaq}. \yong{Combining this and \eqref{eq:deltax} implies that} 
\begin{align*}
\sum_{i=0}^{k-2}L^q_f\|F^q_{i}\|_2\overline{\delta}^{x,q}_{k-1-i} &\leq \mathcal{R}^q\left((\delta^x_0-\frac{\overline{\eta}^q}{1-\theta^q})(k-1){(\theta^q)}^{k-1}+\frac{\overline{\eta}^q}{1-\theta^q}\frac{1-{(\theta^q)}^{k-1}}{1-\theta^q}\right), 
\end{align*} 
\yong{and the upper bound tends to $\mathcal{R}^q \frac{\overline{\eta}^q}{(1-\theta^q)^2}$ as $k$ tends to $\infty$, since} 
$0 < \yong{\theta^q} <1$ (cf. \eqref{eq:thetaq}) and $\lim_{k \to \infty} k\yong{(\theta^q})\moh{^k}=0$ when $0 < \yong{\theta^q} <1$. Moreover, it follows from \yong{the} definitions of $J^q_i$ and $\theta^q$ (cf. Proposition \ref{prop:filterbanks} and Lemma \ref{lem:resdef}), as well as \yong{the} sub-multiplicative \yong{property of norms that:} 
\begin{align*}
\frac{1}{\sqrt{2}}\eta^q_v(\|J^{q,1}_i\|_2+\|J^{q,3}_i\|_2)+\eta^q_w\|J^{q,2}_i\|_2\leq\begin{cases}  O^q, \quad \ \ \ i=0, \\
S^q {\theta^q}^i, \quad i\geq 1,
\end{cases}
\end{align*}
where $ O^q \triangleq \eta^q_w(\|C^q_2\Phi^qG^q_1M^q_1T^q_1\|_2+\|(I-C^q_2G^q_2M^q_2)T^q_2\|_2)+\eta^q_v\|C^q_2\Phi^qW^q\|_2$ and $S^q \triangleq (\eta^q_w\|C^q_2\Phi^qG^q_1M^q_1C^q_1\|_2(\|\Phi^qG^q_1M^q_1T^q_1\|_2+\|G^q_2M^q_2T^q_2\|_2)\hspace{-.1cm}+\hspace{-.1cm}\eta^q_v\|\Phi^qW^q\|_2)$. \yong{Combining this and \eqref{eq:thetaq} results} in
\begin{align*}
\sum_{i=0}^{k-2}\frac{1}{\sqrt{2}}\eta^q_v(\|J^{q,1}_i\|_2+\|J^{q,3}_i\|_2)+\eta^q_w\|J^{q,2}_i\|_2 \leq O^q +S^q \frac{\theta^q-{\theta^q}^{k-1}}{1-\theta^q}, 
\end{align*}
\yong{where the upper bound tends to $\frac{S^q\theta^q}{1-\theta^q}$ as $k$ tends to $\infty$.}
 Next, it is straightforward to observe that \yong{all constitutent} terms in $S^q_{2,k} \triangleq (\|A^q_k\|_2+L^q_f\|F^q_{k-1}\|_2)\delta^x_0+\frac{1}{\sqrt{2}}\eta^q_v(\|J^{q,1}_{k-1}\|_2+\|J^{q,3}_{k-1}\|_2)+\eta^q_w\|J^{q,2}_{k-1}\|_2$ (on the right hand side of \eqref{eq:delta_tri}) are all 
 \yong{decreasing} to zero as $k$ increases, since they are all upper bounded by some 
 \yong{terms involving $({\theta^q})^k$} by their definitions \yong{(cf. Lemma \ref{lem:resdef})} and the sub-multiplicative \yong{property}. 
 Hence, $\lim_{k\to\infty} \delta^{q,tri}_{r,k}= \lim_{k\to\infty} (S^q_{1,k}+  S^q_{2,k})=  \lim_{k\to\infty} S^q_{1,k}< \infty$. 
\end{proof}
 \end{lemma}
 \vspace{-0.3cm}
\begin{lemma} \label{lem:mode_disjointness}
 Suppose that Assumption \ref{assumption:boundedness} holds. 
  Consider two different modes $q \neq q' \in Q$ and their corresponding upper bounds for their residuals' norms, $\delta^{q}_{r,k }$ and $\delta^{q'}_{r,k }$, at time step $k$. At least one of the two modes $q \neq q'$ will be eliminated if
 \begin{align} \label{eq:disjointness}
  &\| C^q_2 \hat{x}^{\star,q}_{k|k}- C^{q'}_2 \hat{x}^{\star,q'}_{k|k}+D^{q}_{2} u^{q}_k-D^{q'}_{2} u^{q'}_k\|_2 >\delta^{q}_{r,k}+\delta^{q'}_{r,k }+R^{q,q'}_z,
  \end{align}
  
\noindent  where $R^{q,q'}_z \triangleq R_y \| T^q_2-T^{q'}_2\|_2$. 
\end{lemma} 
\begin{proof}
Suppose, for contradiction, that none of $q$ and $q'$ are eliminated. Then 
 \begin{align}
 \nonumber &\| C^q_2 \hat{x}^{\star,q}_{k|k}+D^{q}_{2} u^{q}_k- C^{q'}_2 \hat{x}^{\star,q'}_{k|k}-D^{q'}_{2} u^{q'}_k\|_2=\|r^{q'}_{k}-r^{q}_{k}+z^q_{2,k}-z^{q'}_{2,k})\|_2
 \\ \nonumber & \leq \|r^{q'}_{k}\|_2+\|r^{q}_{k}\|_2+\|z^q_{2,k}-z^{q'}_{2,k}\|_2 \leq \delta^{q}_{r,k}+\delta^{q'}_{r,k } + R_y\| T^q_2-T^{q'}_2\|_2, 
 \end{align}
where the equality holds by 
Definition \ref{defn:computedresidual}, the first inequality holds by triangle inequality and the last inequality holds by the assumption that none of $q$ and $q'$ can be eliminated, as well as the boundedness assumption for the measurement space. This 
last inequality contradicts with the inequality in the lemma, thus the result holds.
\end{proof}
\begin{lemma} \label{lem:resdef2}
Consider any mode $q$ with the unknown true mode being $q^{*}$. Suppose without loss of generality \yong{that} $f^q(0)=0$. Then, at time step $k$, we have
\begin{align} \label{eq:resid_comp}
r^{q}_k &=\mathbb{A}^q_kt^q_k+\alpha^{q^*}_k+\epsilon^{q^*}_k, 
 \end{align}
\yong{with $\varepsilon^{q^*}_k$ being} 
an error term that satisfies 
 \begin{align} \label{eq:taylor_error}
\exists \xi_1,\dots, \xi_k \in {X}, \ s.t. \ \|\varepsilon^{q^*}_k\|_2 \leq \frac{1}{2}\sum_{i=1}^k \|J^{q^*}_{f,0}\|^{k-i}_2\|x_{i-1}\|^2_2 \|H^{q^*}_f(\xi_i)\|_2, 
 \end{align}
  where 
 \begin{align*}
 \alpha^{q^*}_k &\triangleq (T^q_2-T^{q^*}_2)( C^{q^*}_{f,k}x_0+C^{q^*}_{d,k}{d}^{q^*}_{0:k}+C^{q^*}_{u,k}{u}^{q^*}_{0:k}+C^{q^*}_{\tilde{w},k}\tilde{w}^{q^*}_{0:k}) \\
 C^{q^*}_{d,k} &\triangleq \begin{bmatrix} H^{q^*} & C^{q^*}G^{q*} & C^{q^*}J^{q^*}_{f,0}G^{q^*} & \dots & C^{q^*}(J^{q^*}_{f,0})^{k-1}G^{q^*}  \end{bmatrix},  \\
 C^{q^*}_{u,k} &\triangleq \begin{bmatrix} D^{q^*} & C^{q^*}B^{q*} & C^{q^*}J^{q^*}_{f,0}B^{q^*} & \dots & C^{q^*}(J^{q^*}_{f,0})^{k-1}B^{q^*}  \end{bmatrix}, \\
 C^{q^*}_{\tilde{w},k} &\triangleq \begin{bmatrix} I & C^{q^*}W^{q*} & C^{q^*}J^{q^*}_{f,0}W^{q^*} & \dots & C^{q^*}(J^{q^*}_{f,0})^{k-1}W^{q^*}  \end{bmatrix}, \\
 {d}^{q^*}_{0:k} &\triangleq \begin{bmatrix} {d}^{q^*\top}_k & \dots & {d}^{q^*\top}_0 \end{bmatrix}^\top, u^{q^*}_{0:k} \triangleq \begin{bmatrix} u^{q*\top}_k  & \dots & u^{q*\top}_0  \end{bmatrix}^{\top},C^{q^*}_{f,k}\triangleq C^{q^*}(J^{q^*}_{f,0})^k, \\
 \tilde{w}^{q^*}_{0:k} &\triangleq \begin{bmatrix} {v}^{q^*\top}_k & {w}^{q^*\top}_{k-1}& \dots & {w}^{q^*\top}_0 \end{bmatrix}^\top, \epsilon^{q^*}_k \triangleq (T^q_2-T^{q^*}_2)\varepsilon^{q^*}_k, 
 \end{align*}
 and $J^{q^*}_{f,0}$ and $H^{q^*}_f(\xi)$ are the Jacobian and Hessian matrices of the vector field $f^{q^*}(\cdot)$ at $0$ and $\xi$, respectively.
 \begin{proof}
 \yong{R}ecall from Proposition \ref{prop:residecomposition}, Lemma \ref{lem:resdef} and \eqref{eq:sys_desc} that:
\begin{align} \label{eq:res_ext}
r^q_k=\mathbb{A}^q_kt^q_k+(T^q_2-T^{q^*}_2)(C^{q^*}x_k+H^{q^*}d^{q^*}_k+D^{q^*}u^{q^*}_k+v^{q^*}_k).
\end{align}
 On the other hand, by applying Taylor series expansion 
 \yong{to} \eqref{eq:sys_desc} we obtain:
\begin{align} \label{eq:taylor}
x_k=J^{q^*}_{f,0}x_{k-1}+B^{q^*}u^{q^*}_{k-1}+G^{q^*}d^{q^*}_{k-1}+W^{q^*}w^{q^*}_{k-1}+(H.O.T)^{q^*}_k,
\end{align}
where 
$(H.O.T)^{q^*}_k$ is an error term that satisfies $\| (H.O.T)^{q^*}_k\|_2 \leq \frac{1}{2}H^{q^*}_f(\xi_k)$ for some $\xi_k \in X$. \yong{Then, by a}pplying \eqref{eq:taylor} at time steps $k,k-1,\dots,1$, plugging them in\yong{to} \eqref{eq:res_ext} and augmentat\yong{ing the results, we obtain} 
\eqref{eq:resid_comp}.
 \end{proof}
\end{lemma}
\vspace{-0.1cm}
\begin{theorem}[Sufficient Conditions for Mode Detectability] \label{thm:strong_mode_detect}
System \eqref{eq:sys_desc} is mode detectable, i.e., by applying Algorithm \ref{algorithm1}, all false modes will be eliminated \yong{at some large enough time step $K$}, 
if the assumptions in Proposition \ref{prop:filterbanks} and either of the following hold: 
\renewcommand{\theenumi}{\roman{enumi}}
\begin{enumerate}
\item Assumption \ref{assumption:boundedness} \yong{holds} and $ \forall q,q' \in Q$, $q\neq q'$, \label{item:second}
\begin{align*}
  \sigma_{min} (W^{q,q'}) > \frac{\overline{\delta}^{q,tri}_{r}+\overline{\delta}^{q',tri}_{r}+R^{'q,q'}_y}{\sqrt{R^2_x+\eta^2_v}}, 
\end{align*}
where $W^{q,q'}\hspace{-0.1cm} \triangleq \hspace{-0.1cm}\begin{bmatrix} (C^q_2 - C^{q'}_2) & (T^q_2 - T^{q'}_2) & -I  & I & D^q_2  & -D^{q'}_2 \end{bmatrix}$.
\item \label{item:first} Assumption \ref{as:2} \yong{holds} and {$T^q_2 \neq T^{q'}_2$ holds $ \forall q,q' \in Q, q \neq q'$}. Moreover, $H^{q^*}_f(\cdot)$ is bounded on $X$ and $\|J^{q^*}_{f,0}\|_2 <1$.  
\end{enumerate} 
\end{theorem}
\begin{proof}
To show that \eqref{item:second} is sufficient for asymptotic mode detectability, consider Lemma \ref{lem:mode_disjointness} with $\delta^{q,tri}_{r,k}$ as the upper bound. It suffices to show that $\exists K \in \mathbb{N}$, such that \eqref{eq:disjointness} holds for $k \geq K, \forall q \neq q' \in \mathbb{Q}.$ 
 Notice that by Definition \ref{defn:computedresidual}, $C^q_2 \hat{x}^{\star,q}_{k|k}=C^q_2 x_k +T^q_2 v_k - r^{q|*}_k$.
 \yong{Hence, by  p}lugging this into \eqref{eq:disjointness}, we need to show \yong{that} $\exists K \in \mathbb{N}$ such that: 
\begin{align}  \label{eq:sufficient}
\begin{array}{ll}
 &\| W^{q,q'} s^{q,q'}_k \|_2 >\delta^{q,tri}_{r,k}+\delta^{q',tri}_{r,k}+R^{q,q'}_z, \forall k \geq K, \forall q \neq q' \in \mathbb{Q},
\end{array}
\end{align} 
\yong{where $s^{q,q'}_k \triangleq \begin{bmatrix} x^{\top}_k & v^{\top}_k & r^{q|* \top}_k & r^{q'|* \top}_k & u^{q \top}_k & u^{q' \top}_k \end{bmatrix}^{\top}$.}
A sufficient condition to satisfy \eqref{eq:sufficient} is that $\exists K \in \mathbb{N}$ such that $\forall k \geq K$, \eqref{eq:sufficient} holds for all $s^{q,q'}_k$. Equivalently, it suffices that:
$$\yong{\underline{W}^{q,q'}_k}>\delta^{q,tri}_{r,k}+\delta^{q',tri}_{r,k}+R^{q,q'}_z, \yong{\forall k \geq K, \forall q \neq q' \in \mathbb{Q},}$$
\yong{where 
 \begin{align}
 \nonumber \underline{W}^{q,q'}_k \triangleq&\min \limits_{ x_k,v_k,r^q_k,r^{q'}_k } \| W^{q,q'} s^{q,q'}_k \|_2
 \\ \nonumber & s.t. \ \|x_k\|_2 \leq R_x, \|v_k\|_2 \leq \eta_v, \|r^{q|*}_k\|_2 \leq \delta^{q,tri}_{r,k}, \|r^{q'|*}_k\|_2 \leq \delta^{q',tri}_{r,k}. 
 \end{align}
 Finally, b}y expanding the constraint set, it suffices to require that $\exists K \in \mathbb{N}$ such that:
 $$\yong{\underline{\underline{W}}^{q,q'}_k}>\delta^{q,tri}_{r,k}+\delta^{q',tri}_{r,k}+R^{q,q'}_z, \yong{\forall k \geq K, \forall q \neq q' \in \mathbb{Q},}$$
 where
  \begin{align}
 \nonumber \yong{\underline{\underline{W}}^{q,q'}_k\triangleq} &\min \limits_{s^{q,q'}_k} \| W^{q,q'} s^{q,q'}_k \|_2 
 \\ \nonumber s.t. \ &\| s^{q,q'}_k\|^2_2 \leq R^2_x+ \eta^2_v + (\delta^{q,tri}_{r,k})^2  + (\delta^{q',tri}_{r,k})^2+(u^q_k)^2+(u^{q'}_k)^2. 
 \end{align}
 Now, by \yong{the} \emph{matrix lower bound} theorem {\cite{grcar2010matrix}} and \yong{a} similar argument 
 \yong{to} the proof of Lemma \ref{lem:delta_inf_diverge}, it is sufficient to \yong{require} 
 that $ \exists K \in \mathbb{N}$ \yong{such that} 
 $\forall k \geq K, \forall q \neq q' \in \mathbb{Q} : $ 
 \begin{align} \label{eq:timed_sufficient_compatibility} 
 \sigma^2_{min} (W^{q,q'}) \hspace{-0.1cm}>\hspace{-0.1cm} \frac{ (\delta^{q,tri}_{r,k}+ \delta^{q',tri}_{r,k}+R^{q,q'}_z)^2}{R^2_x\hspace{-0.1cm}+\hspace{-0.1cm}\eta^2_v\hspace{-0.1cm}+\hspace{-0.1cm}( \delta^{q,tri}_{r,k})^2\hspace{-0.1cm}+\hspace{-0.1cm}( \delta^{q',tri}_{r,k})^2\hspace{-0.1cm}+\hspace{-0.1cm}(u^q_k)^2\hspace{-0.1cm}+\hspace{-0.1cm}(u^{q'}_k)^2}.
 \end{align}
 The result in \eqref{eq:timed_sufficient_compatibility} provides us a \emph{time-dependent} sufficient condition for mode detectability. In order to find a \emph{time-independent} sufficient condition, notice that $ \frac{ (\overline{\delta}^{q,tri}_{r,k}+ \overline{\delta}^{q',tri}_{r,k}+R^{q,q'}_z)^2}{R^2_x+\eta^2_v}$ is an upper bound for the right hand side of \eqref{eq:timed_sufficient_compatibility}, since the latter's denominator is smaller than the former's and the numerator of the latter is an upper bound signal for the former's  
  by triangle \yong{inequality} and \yong{the} sub-multiplicative \yong{property of norms}. 
  So, a sufficient condition for \eqref{eq:timed_sufficient_compatibility} is \yong{that} $ \exists K \in \mathbb{N}$ \yong{such that} $\forall k \geq K, \forall q \neq q' \in \mathbb{Q} : $
 \vspace{-0.1cm} \begin{align} \label{eq:timed_sufficient_compatibility_2}
 \sigma^2_{min} (W^{q,q'}) > \frac{ (\overline{\delta}^{q,tri}_{r,k}+ \overline{\delta}^{q',tri}_{r,k}+R^{q,q'}_z)^2}{R^2_x+\eta^2_v}.
 \end{align}
Then, for the above to hold, it suffices that 
 $$\sigma^2_{min} (W^{q,q'}) > \lim_{k\to \infty} \frac{ (\overline{\delta}^{q,tri}_{r,k}+ \overline{\delta}^{q',tri}_{r,k}+R^{q,q'}_z)^2}{R^2_x+\eta^2_v},$$
 which is equivalent to \eqref{item:second} by \eqref{eq:r_tri_converge}.
 
As for the sufficiency of \eqref{item:first}, we show that the sufficient conditions in \eqref{item:first} imply that if $q \ne q^*$, then the residual signal $r^q_k$ \yong{grows}  unbounded. Then, since we showed in Lemma \ref{lem:delta_inf_diverge} that the computed upper bound signal $\hat{\delta}^q_{r,k}$ is bounded, so there must exist a time step $K$ such that $r^q_k > \hat{\delta}^q_{r,k}$ for $k \geq K$, and hence, mode $q$ will be eliminated after time step $K$ and therefor\yong{e,} mode detectability holds. To do so, we show that if \eqref{item:first} holds, then the right hand side of \eqref{eq:resid_comp} \yong{grows}  unbounded, and so 
\yong{does} $r^q_k$. First, note that by Lemma \ref{lem:delta_inf_diverge}, the first term in the right hand side of \eqref{eq:resid_comp}, i.e., $\mathbb{A}^q_kt^q_k$, is bounded. Moreover, \eqref{eq:taylor_error} and the facts that the state space is bounded and $\|J^{q^*}_{f,0}\|_2 <1$ imply that $\epsilon^{q^*}_k$, i.e., the third term in the right hand side of \eqref{eq:resid_comp}, \yong{is bounded}. 
  
Next, we show that the second term in the right hand side of \eqref{eq:resid_comp}, i.e.  $\alpha^{q^*}_k$, \yong{grows}  unbounded. Consequently, the summation of the two bounded terms $\mathbb{A}^q_kt^q_k$ and $\epsilon^{q^*}_k$ as well as the unbounded term $\alpha^{q^*}_k$ will be unbounded. To show that $\alpha^{q^*}_k$ \yong{grows}  unbounded,  
it suffices to show that for any $c>0$, any specific mode $q$ with the true mode being  $q^*$, there exists \yong{a} large enough $K$ such that:
\begin{align}
\nonumber \|\alpha^{q^*}_K\|_2=\left\|\begin{bmatrix} \mathbb{T}^{q,q^*}_K & \mathbb{C}^{q,q^*}_{u,K} & \mathbb{C}^{q,q^*}_{d,K}  \end{bmatrix} \begin{bmatrix} \zeta^{\top}_K & u^{q^*\top}_{0:k} & d^{q*\top}_{0:K} \end{bmatrix}^{\top}\right\|_2>c,
 \end{align}
{with $ \mathbb{T}^{q,q^*}_K \triangleq (T^{q}-T^{q^*})\begin{bmatrix} C^{q^*}_{x,K} & C^{q^*}_{\tilde{w},K} \end{bmatrix}$, $\mathbb{C}^{q,q^*}_{u,K} \triangleq (T^{q}-T^{q^*}){C}^{q^*}_{u,K}$, $\mathbb{C}^{q,q^*}_{d,K} \triangleq (T^{q}-T^{q^*}){C}^{q^*}_{d,K}$ and $\zeta_K \triangleq \begin{bmatrix} x^\top_0 & {\tilde{w}^{q^*\top}_{0:K}} \end{bmatrix}^\top $.} Since $q^*$ is unknown, a sufficient condition {to satisfy} 
the above equality is \yong{that} $ \forall c>0, \forall q' \neq q \in Q, \exists K \in \mathbb{N}  $ \yong{such that:}
 \begin{align}
\nonumber \left\|\begin{bmatrix} \mathbb{T}^{q,q'}_K & \mathbb{C}^{q,q'}_{u,K} & \mathbb{C}^{q,q'}_{d,K}  \end{bmatrix} \begin{bmatrix} \zeta^{\top}_K & u^{q'\top}_{0:K} & d^{q*\top}_{0:k} \end{bmatrix}^{\top}\right\|_2>c.
 \end{align}
So it suffices that $\forall c>0, \forall q' \neq q \in Q, \exists \overline{d} \in \mathbb{R}, \exists K \in \mathbb{N}$, such that: 
\yong{$$\underline{T}^{q,q'}_k>c,$$
where}
  \begin{align}
 \nonumber \yong{\underline{T}^{q,q'}_k \triangleq} &\min \limits_{\zeta'_k} \left\|\begin{bmatrix} \mathbb{T}^{q,q'}_K & \mathbb{C}^{q,q'}_{u,K} & \mathbb{C}^{q,q'}_{d,K}  \end{bmatrix} {\zeta'_K}\right\|_2
 \\[-0.2cm] \nonumber &s.t. \ \zeta'_K=\begin{bmatrix} x^\top_0 & {\tilde{w}^{q^*\top}_{0:K}} & u^{q'\top}_{0:K} & d^{q*\top}_{0:K} \end{bmatrix}^{\top},\|d^{q*}_{0:K}\|_2 \geq \overline{d},
 \\[-0.05cm] \nonumber &\ \ \ \ \  \|w_i\|_{\infty} \leq \eta_w, \ \| v_j \|_{\infty} \leq \eta_v,  i \in \{ 0,...,K-1 \},  j \in \{ 0,...,K \}{.} 
  \end{align}
\yong{Once again, by the} matrix lower bound theorem, a sufficient condition  for the above inequality to hold is that $\exists \overline{d} \in \mathbb{R}, \exists K \in \mathbb{N}$, such that: \vspace{-0.15cm}
\yong{$$\underline{\underline{T}}^{q,q'}_k>\frac{c}{\sigma_{min}(\begin{bmatrix}  \mathbb{T}^{q,q'}_K & \mathbb{C}^{q,q'}_{u,K} & \mathbb{C}^{q,q'}_{d,K}  
 \end{bmatrix})},$$
 where}
  \begin{align} \label{eq:suff}
  \yong{\underline{\underline{T}}^{q,q'}_k\triangleq} &\min \limits_{\tilde{w}^{q^*}_{0:K},d^{q^*}_{0:K}}\|\zeta'_K\|_2  
 \\[-0.2cm] \nonumber &s.t. \ \zeta'_K=\begin{bmatrix} x^\top_0 & {\tilde{w}^{q^*\top}_{0:K}} & u^{q'\top}_{0:k} & d^{q*\top}_{0:K} \end{bmatrix}^{\top},\|d^{q*}_{0:K}\|_2 \geq \overline{d},
 \\ \nonumber  &\ \ \ \ \   \ \|w_i\|_{\infty} \leq \eta_w, \ \| v_j \|_{\infty} \leq \eta_v, i \in \{ 0,...,K-1 \},  j \in \{ 0,...,K \}.
  \end{align}
  Finally, since 
  \begin{align}
  \nonumber \|\zeta'_K\|_2 =\left\| \begin{bmatrix} x^\top_0 & {\tilde{w}^{q^*\top}_{0:k}} & u^{q'\top}_{0:K} & d^{q*\top}_{0:K} \end{bmatrix}\right\|_2 \geq \sqrt{0^2+0^2+0^2+\|d^{q*\top}_{0:K}\|^2_2}=\| d^{q*\top}_{0:K}\|_2,
  \end{align}
 then a sufficient condition for \eqref{eq:suff} \yong{to hold} is that 
   \begin{align} \label{eq:suff2}
  \|d^{q*\top}_{0:K}\|_2 >\frac{c}{\sigma_{min}{(}\begin{bmatrix}  \mathbb{T}^{q,q'}_K & \mathbb{C}^{q,q'}_{u,K} & \mathbb{C}^{q,q'}_{d,K}\end{bmatrix}{)}}.
  \end{align}
Now, suppose that $T^q_2 \neq T^{q'}_2$ (otherwise the matrix in the denominator of \eqref{eq:suff2} is zero and it never holds). \yong{Then,} 
the right hand side of \eqref{eq:suff2} converges \yong{asymptotically} to $\tilde{\delta} \triangleq \max \{0, \frac{c}{\overline{\sigma}^{q,q'}} \}$, since 
the 
\yong{smallest} singular value in the denominator either diverges, or converges to some steady value $\overline{\sigma}^{q,q'}$. So we set $\overline{d}$ \yong{to be} equal to any real number \yong{that is strictly greater} than $\tilde{\delta}$. By \yong{the} unlimited energy assumption for 
\yong{the unknown input} signal, 
\yong{at} some large enough time step $K$, the monotone\yong{ly} increasing function $\|d^{q*}_{0:k}\|_2$ \yong{will exceed} $\overline{d}$ and so, the system will be mode detectable. 
\end{proof}
\section{Simulation Results} \label{sec:examples}
In this section, we evaluate the effectiveness of our Simultaneous Mode, Input, and State Set-Valued Observer (SMIS), by comparing its performance with 
\yong{the} LMI-based $\mathcal{H}_{\infty}$-
observer in \cite{zheng2020asynchronous} that 
obtains point state estimates. For comparison, we apply the two observers on a modified version of the discrete-time nonlinear switched system in \cite{zheng2020asynchronous}, 
where we increase the number of modes (subsystems) to five, i.e., $Q=5$. The considered system is in the form of \eqref{eq:sys_desc}, with the following parameters: $n=l=2, m=p=1$ and $\forall q=1,\dots,5$: 
\begin{align*}
B^q=D^q=0_{2 \times 1}, \ f^q(x)=\tilde{A}^q\gamma(x)+\hat{A}^qx, 
\end{align*}
where $\gamma(x) \triangleq \frac{1}{2} \begin{bmatrix} \sin(x_1) & \sin(x_2) \end{bmatrix}^\top$. Moreover, 
\begin{align*}
\hat{A}^1\hspace{-.1cm}&=\hspace{-.1cm}\begin{bmatrix} 0.3 & 0 \\ 0.4 & -0.7 \end{bmatrix}\hspace{-.1cm},\tilde{A}^1\hspace{-.1cm}=\hspace{-.1cm}\begin{bmatrix} 0.8 & -0.4 \\ 0.4 & -0.8 \end{bmatrix}\hspace{-.1cm}, C^1\hspace{-.1cm}=\hspace{-.1cm}\begin{bmatrix} 0.8 & 0.1 \\ 0.8 & 0.1 \end{bmatrix}\hspace{-.1cm}, H^1\hspace{-.1cm}=\hspace{-.1cm}\begin{bmatrix} 0.5 \\ 0.5 \end{bmatrix}\hspace{-.1cm}, G^1\hspace{-.1cm}=\hspace{-.1cm}\begin{bmatrix} 0.4 \\ -0.1 \end{bmatrix}\hspace{-.1cm},\\
\hat{A}^2\hspace{-.1cm}&=\hspace{-.2cm}\begin{bmatrix} -0.5 & 0 \\ 1 & -0.5 \end{bmatrix}\hspace{-.1cm},\tilde{A}^2\hspace{-.1cm}=\hspace{-.2cm}\begin{bmatrix} 0.6 & -0.1 \\ 0.1 & -0.6 \end{bmatrix}\hspace{-.1cm}, C^2\hspace{-.1cm}=\hspace{-.2cm}\begin{bmatrix} 0.5 & -0.1 \\ 0.6 & -0.1 \end{bmatrix}\hspace{-.1cm}, H^2\hspace{-.1cm}=\hspace{-.2cm}\begin{bmatrix} 0.6 \\ -0.5 \end{bmatrix}\hspace{-.1cm}, G^2\hspace{-.1cm}=\hspace{-.2cm}\begin{bmatrix} -0.2 \\ 0.1 \end{bmatrix}\hspace{-.1cm},\\
\hat{A}^3\hspace{-.1cm}&=\hspace{-.2cm}\begin{bmatrix} 0.6 & -0.2 \\ -0.4 & 0.7 \end{bmatrix}\hspace{-.1cm},\tilde{A}^3\hspace{-.1cm}=\hspace{-.2cm}\begin{bmatrix} 0.4 & -0.8 \\ -0.2 & -0.4 \end{bmatrix}\hspace{-.1cm}, C^3\hspace{-.1cm}=\hspace{-.2cm}\begin{bmatrix} 0.2 & 0.7 \\ -0.8 & 0.2 \end{bmatrix}\hspace{-.1cm}, H^3\hspace{-.1cm}=\hspace{-.2cm}\begin{bmatrix} -0.5 \\ 0.5 \end{bmatrix}\hspace{-.1cm}, G^3\hspace{-.1cm}=\hspace{-.2cm}\begin{bmatrix} 0.5 \\ 0.2 \end{bmatrix}\hspace{-.1cm},\\
\hat{A}^4\hspace{-.1cm}&=\hspace{-.2cm}\begin{bmatrix} -0.6 & -0.2 \\ 0.4 & 0.7 \end{bmatrix}\hspace{-.1cm},\tilde{A}^4\hspace{-.1cm}=\hspace{-.2cm}\begin{bmatrix} -0.4 & 0.9 \\ 0.2 & -0.3 \end{bmatrix}\hspace{-.1cm}, C^4\hspace{-.1cm}=\hspace{-.2cm}\begin{bmatrix} 0.3 & -0.7 \\ 0.8 & -0.6 \end{bmatrix}\hspace{-.1cm}, H^4\hspace{-.1cm}=\hspace{-.2cm}\begin{bmatrix} -0.4 \\ 0.9 \end{bmatrix}\hspace{-.1cm}, G^4\hspace{-.1cm}=\hspace{-.2cm}\begin{bmatrix} 0.9 \\ 0.3 \end{bmatrix}\hspace{-.1cm},\\
\hat{A}^5\hspace{-.1cm}&=\hspace{-.2cm}\begin{bmatrix} -0.2 & 0.9 \\ -0.1 & 0.3 \end{bmatrix}\hspace{-.1cm},\tilde{A}^5\hspace{-.1cm}=\hspace{-.2cm}\begin{bmatrix} -0.8 & 0.1 \\ 0.3 & -0.7 \end{bmatrix}\hspace{-.1cm}, C^5\hspace{-.1cm}=\hspace{-.2cm}\begin{bmatrix} -0.3 & -0.1 \\ -0.8 & 1 \end{bmatrix}\hspace{-.1cm}, H^5\hspace{-.1cm}=\hspace{-.2cm}\begin{bmatrix} -0.1 \\ 0.1 \end{bmatrix}\hspace{-.1cm}, G^5\hspace{-.1cm}=\hspace{-.2cm}\begin{bmatrix} 0.6 \\ 0.1 \end{bmatrix}\hspace{-.1cm}.
\end{align*}
The initial state estimate and noise signals \yong{satisfy $\|x_0\|\yong{_2}\le \delta_x=0.5$, $\|w_k\|\yong{_2}\le \eta_w=0.02$ and $\|v_k\|\yong{_2}\le \eta_w=0.02$.} 
Furthermore, we assume \yong{that} $\hat{x}_{0|0}=\begin{bmatrix} 0.4 & 0.4 \end{bmatrix}^\top$. 

We consider two scenarios for the unknown input. In the first (\yong{Scenario I}), the unknown input is a random signal with bounded norm, i.e., $\|d_k\|_2 \leq 0.4$, 
while $d_k$ in the second scenario (\yong{Scenario II}) is a time-varying signal that becomes unbounded as time increases. As is demonstrated in Figure \ref{fig:estimates_scen1}, in the first scenario, i.e., with bounded unknown inputs, the set estimates of our approach (i.e., SMIS estimates) converge to steady-state values and the point estimates of the 
approach in \cite{zheng2020asynchronous} are within the predicted upper bounds and exhibit convergent behavior. 
\begin{figure}[!hp]
\begin{center}
\includegraphics[scale=0.19,trim=20mm 3mm 20mm 5mm,clip]{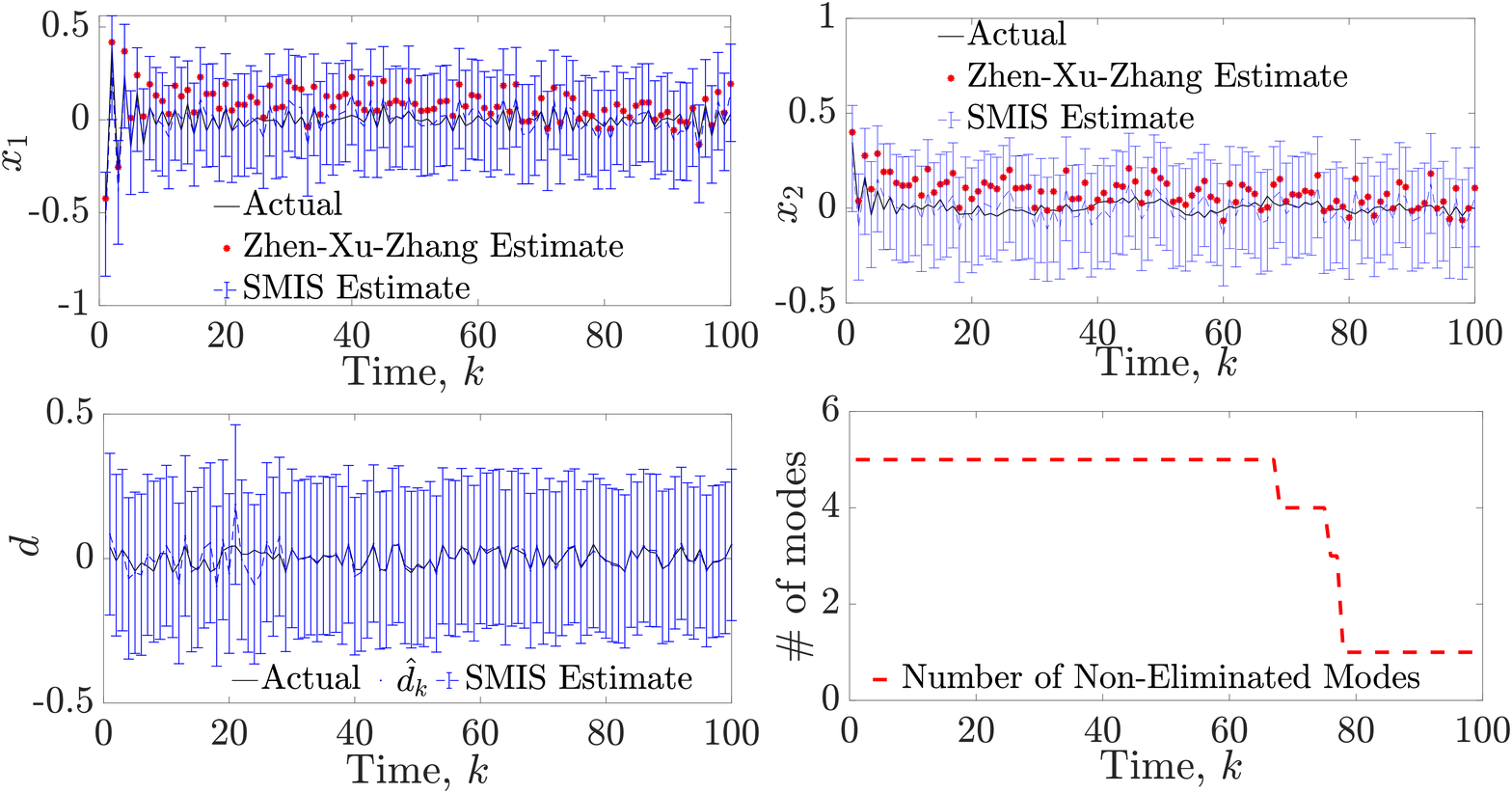}
\caption{Actual states $x_1$, $x_2$, and their estimates, as well as \yong{the} unknown input $d$ and its estimates, and the number of non-eliminated modes \yong{at each time step} in the bounded unknown input scenario \yong{(Scenario I), when} 
applying the observer in \cite{zheng2020asynchronous} (Zhen-Xu-Zhang Estimate) and our {proposed} 
observer (SMIS Estimate).}\label{fig:estimates_scen1}
\includegraphics[scale=0.195,trim=20mm 8mm 20mm 5mm,clip]{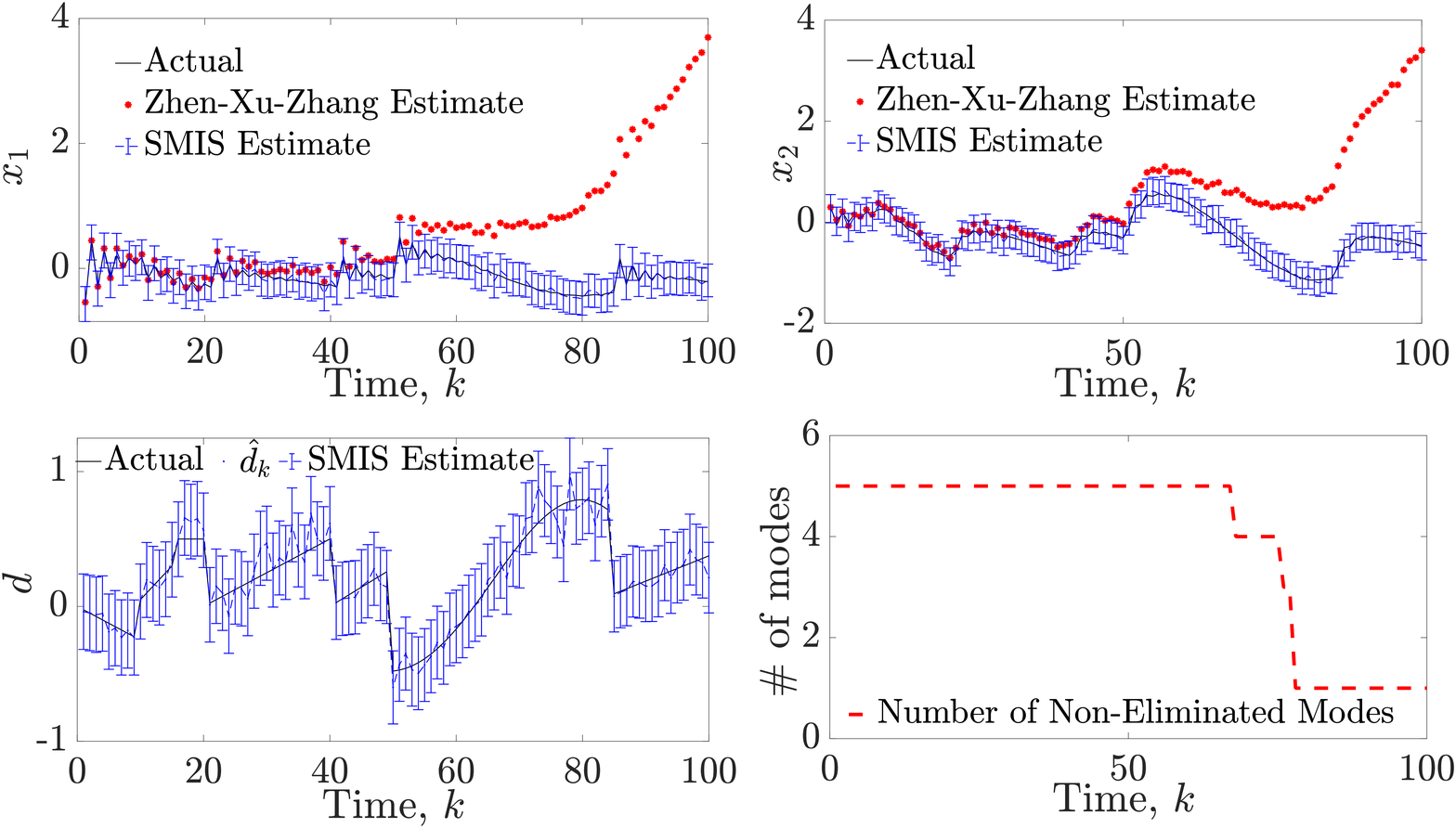}
\caption{Actual states $x_1$, $x_2$, and their estimates, as well as \yong{the} unknown input $d$ and its estimates, and the number of non-eliminated modes 
\yong{at each time step} in the unbounded unknown input scenario \yong{(Scenario II), when} 
applying the observer in \cite{zheng2020asynchronous} (Zhen-Xu-Zhang Estimate) and our {proposed} 
observer (SMIS Estimate). \label{fig:estimates_scen2} }
\end{center}
\end{figure}
More interestingly, considering the second scenario, i.e., with unbounded unknown inputs, Figure \ref{fig:estimates_scen2} shows that our set-valued estimates still converge, i.e., our observer remains stable, while the estimates of the  approach \yong{in \cite{zheng2020asynchronous}} exceed the boundaries of the compatible sets of states and inputs \yong{of our approach} after some time steps  and display a divergent behavior 
(cf. Figure \ref{fig:estimates_scen2}).

\begin{table}[th] 
	\centering \setlength{\tabcolsep}{3pt}
	\caption{Different modes and their $T^q_2$ \yong{in Scenario I (i.e., with bounded $d_k$)}. 
	\label{fig:mode_table}} 
	\begin{tabular}{| c | c |}
		\hline
		  Mode & $T_2^q$ \\ \hline
		$q=1$  & [0.3629  -0.2179 ]$^\top$ \\ \hline
		$q=2$ &  [0.1191  0.8715  ]$^\top$ \\ \hline
		$q=3$ &  [-0.6468  0.8390  ]$^\top$ \\ \hline
		$q=4$ &  [0.8103  -0.6681  ]$^\top$ \\ \hline
		$q=5$ &  [0.2780  -0.6793  ]$^\top$ \\ \hline
	\end{tabular} \normalsize 
	\centering \setlength{\tabcolsep}{3pt}
	\caption{Different modes and their $T^q_2$ \yong{in Scenario II (i.e., with unbounded $d_k$)}. 
	\label{fig:mode_table_2}} 
	\begin{tabular}{| c | c |}
		\hline
		  Mode & $T_2^q$ \\ \hline
		$q=1$  & [0.4730  -0.3280 ]$^\top$ \\ \hline
		$q=2$ &  [0.2202  0.9826  ]$^\top$ \\ \hline
		$q=3$ &  [-0.7579  0.9401  ]$^\top$ \\ \hline
		$q=4$ &  [0.9214  -0.7792  ]$^\top$ \\ \hline
		$q=5$ &  [0.3891  -0.7804  ]$^\top$ \\ \hline
	\end{tabular} \normalsize \vspace{-0.1cm}
\end{table}

\yong{Further,} Tables \ref{fig:mode_table} and \ref{fig:mode_table_2} 
\yong{show} the matrix $T^q_2$ for each mode $q$ \yong{for Scenarios I and II,} 
respectively. \yong{It can be verified that} 
the {second} 
set of sufficient conditions in Theorem \ref{thm:strong_mode_detect} holds, i.e., {$T^q_2 \neq T^{q'}_2$ for all $q \neq q'$,} for both scenarios. 
Hence, we expect that all false modes are eliminated, i.e., 
\yong{exactly} one (true) mode remains, \yong{after some large enough time in} 
both scenarios, which \yong{is indeed what we observe} 
in Figures \ref{fig:estimates_scen1} and \ref{fig:estimates_scen2}, where the number of non-eliminated modes at each time step is \yong{shown}. 

Moreover, for each 
mode $q$, the signals $\|r^q_{k}\|_2$, $\|r^{q|*}_{k}\|_2$, $\delta^{q,tri}_{r,k}$ and $\delta^{q,inf}_{r,k}$ are depicted in Figures \ref{fig:residuala_scen1} and \ref{fig:residuala_scen2} for \yong{Scenarios I and II}, respectively. 
\yong{In both scenarios, we observe that $\delta^{q,inf}_{r,k}$ is smaller than $\delta^{q,tri}_{r,k}$ up until approximately 40 time steps, after which $\delta^{q,tri}_{r,k}$ is smaller/tighter. This is one of the main reasons we considered the minimum of both as the threshold in our mode elimination algorithm (also see Remark \ref{rem:1}).} 
Furthermore, for all modes, $\delta^{q,tri}_{r,k}$ is eventually  convergent while $\delta^{q,inf}_{r,k}$ diverges, as 
\yong{proven} in Lemma \ref{lem:delta_inf_diverge}. So, after some large enough time, $\delta^{q,tri}_{r,k}$ can be used as our upper bound \yong{threshold}, while $\delta^{q,inf}_{r,k}$ becomes \yong{ineffective}. 
 \begin{figure}[!tp]
\begin{center}
\includegraphics[scale=0.2306,trim=33mm 3mm 20mm 5mm,clip]{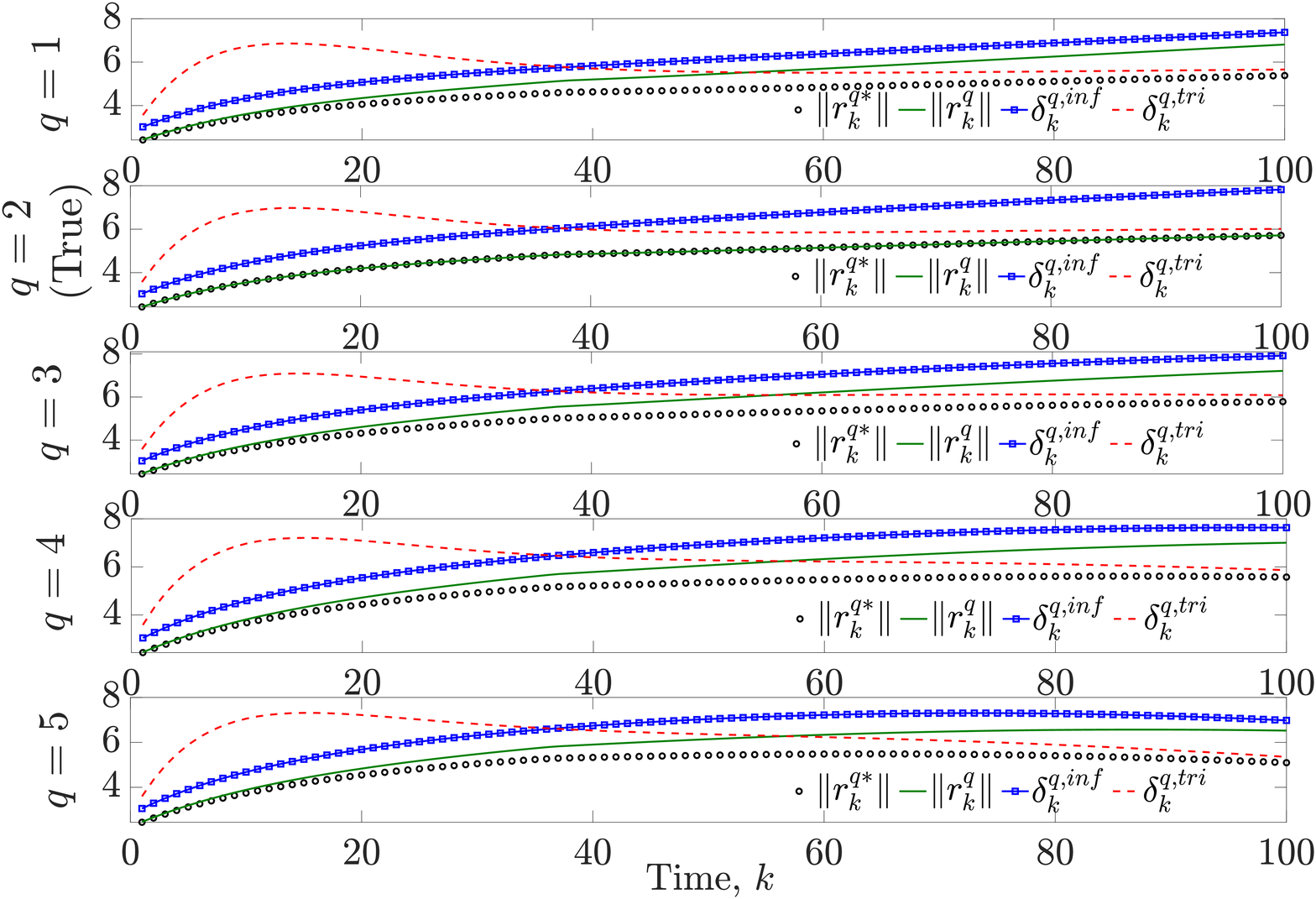}
\caption{$\|r^{q}_{r,k}\|_2$,$\|r^{q|*}_{r,k}\|_2$ and their upper bounds for different modes in the bounded unknown input scenario \yong{(Scenario I).} 
 \label{fig:residuala_scen1} }
\includegraphics[scale=0.237506,trim=33mm 3mm 20mm 5mm,clip]{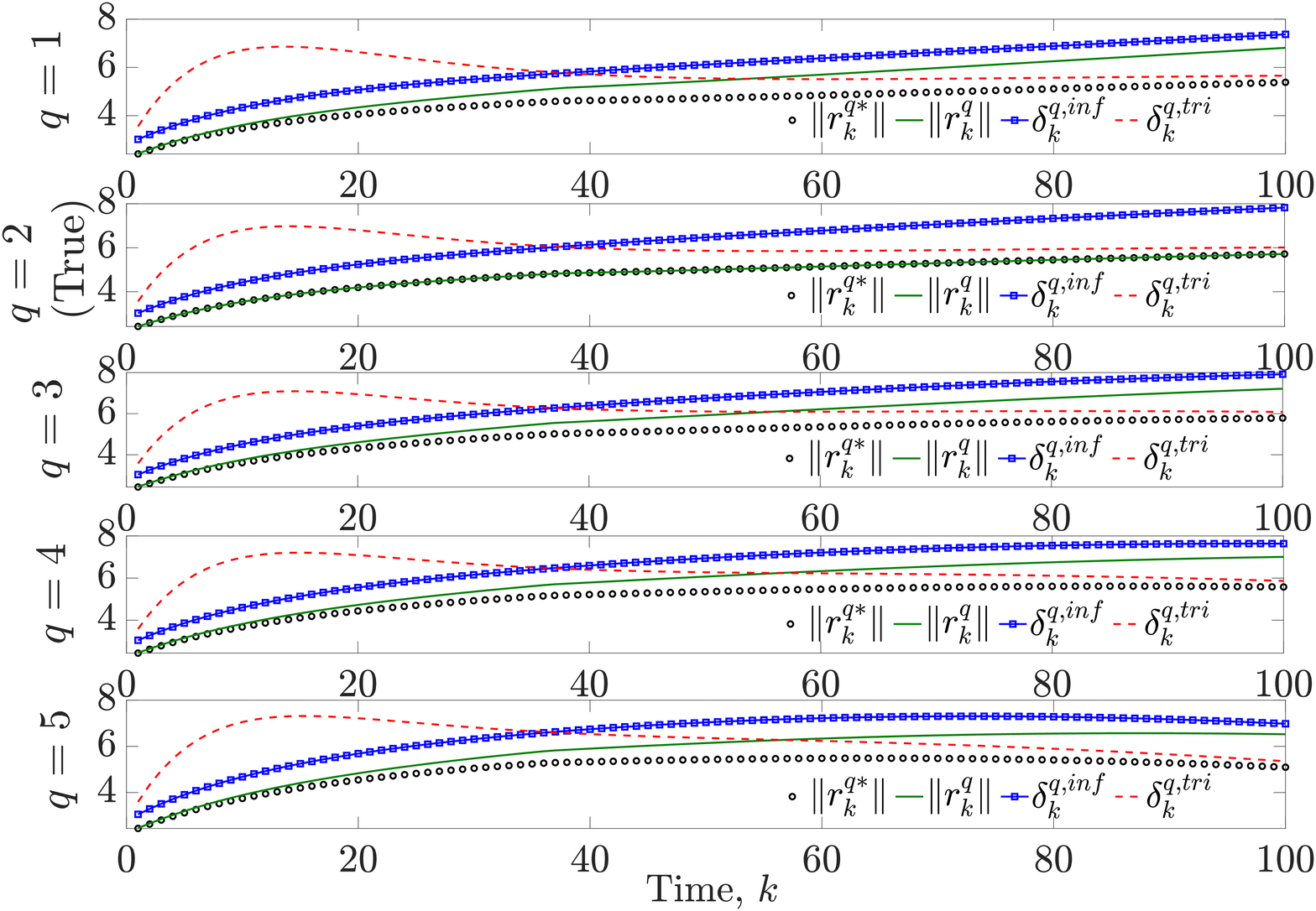}
\caption{$\|r^{q}_{r,k}\|_2$,$\|r^{q|*}_{r,k}\|_2$ and their upper bounds for different modes in the unbounded unknown input scenario \yong{(Scenario II).} 
 \label{fig:residuala_scen2} }
\end{center}
\end{figure}
\section{Conclusion} \label{sec:conclude}
\yong{This paper introduced a novel multiple-model approach for simultaneous mode, unknown input and state estimation for hidden mode switched nonlinear systems with bounded-norm noise and unknown inputs. The proposed approach consists of a bank of mode-matched state and unknown input observer that is optimal in the $\mathcal{H}_\infty$ sense and a mode observer, which eliminates inconsistent modes and their corresponding observers at each time step. The proposed mode elimination criterion is based on the use of a provably finite-valued upper bound for the norm of a residual signal as the threshold. Moreover, we provided a tractable approach to compute the threshold signal and proved}
the convergence of the upper bound\yong{/threshold} signal 
		\yong{as well as} derived sufficient conditions for eventually eliminating all false modes \yong{when} {using} our {mode elimination algorithm}. 
		Finally, we demonstrated  the effectiveness of our {observer} using an illustrative example, where we compared our approach with \yong{an existing $\mathcal{H}_\infty$ observer in the literature under two different scenarios.} 

\section*{Acknowledgments}

This work is partially supported by the National Science Foundation, USA grants CNS-1943545 and CNS-1932066.

\section*{References}
\bibliography{biblio,fault}


\appendix 
\noindent \section*{Appendices}
 \section{System Transformation \yong{\cite[Appendix A.1]{khajenejad2020simultaneousnonlinear}} } \label{app:transformation}
For $q \in \mathbb{Q}$, let $p_{H^q}\triangleq {\rm rk} (H^q)$. Using singular value decomposition, we rewrite the direct feedthrough matrix $H^q$  as
$H^q= \begin{bmatrix}U^q_{1}& U^q_{2} \end{bmatrix} \begin{bmatrix} \Sigma^q & 0 \\ 0 & 0 \end{bmatrix} \begin{bmatrix} V_{1}^{q \top} \\ V_{2}^{q \top} \end{bmatrix}$, 
where $\Sigma^q \in \mathbb{R}^{p_{H^q} \times p_{H^q}}$ is a diagonal matrix of full rank, $U^q_{1} \in \mathbb{R}^{l \times p_{H^q}}$, $U^q_{2} \in \mathbb{R}^{l \times (l-p_{H^q})}$, $V^q_{1} \in \mathbb{R}^{p \times p_{H^q}}$ and $V^q_{2} \in \mathbb{R}^{p \times (p-p_{H^q})}$, while $U^q\triangleq \begin{bmatrix} U^q_{1} & U^q_{2} \end{bmatrix}$ and $V^q\triangleq \begin{bmatrix} V^q_{1} & V^q_{2} \end{bmatrix}$ are unitary matrices. 
When there is no direct feedthrough, $\Sigma^q$, $U^q_{1}$ and $V^q_{1}$ are empty matrices\footnote{\ Based on the convention that the inverse of an empty matrix is an empty matrix and the assumption that operations with empty matrices are possible.}, 
and $U^q_{2}$ and $V^q_{2}$ are arbitrary unitary matrices, {while when $p_{H^q}=p=l$, $U^q_{2}$ and $V^q_{2}$ are empty matrices, and $U^q_{1}$ and $\Sigma^q$ are identity matrices}.
Then, 
we decouple the unknown input into two orthogonal components and since $V^q$ is unitary, we obtain: 
\begin{align}\label{eq:dec}
d^q_{1,k}=V_{1}^{q\top} d^q_k, \quad
d^q_{2,k}=V_{2}^{q\top} d^q_k, \quad d^q_k =V^q_{1} d^q_{1,k}+V^q_{2} d^q_{2,k}.
\end{align}
 So, we can represent system \eqref{eq:sys_desc} as:
\begin{align}
\hspace{-0.3cm}\begin{array}{rl} x^q_{k+1}
&= f^q(x_k) +B^qu^q_k+ G^q_{1} d^q_{1,k} +G^q_{2} d^q_{2,k}+W^qw^q_k, 
\\  y_k
&={C}^q x_k + D^q_ku^q_k+ H^q_{1} d^q_{1,k}+v^q_k, \end{array}\hspace{-0.3cm}   \label{eq:y}
\end{align}
where $G^q_{1} \triangleq G^q V^q_{1}$, $G^q_{2} \triangleq G^q V^q_{2}$ and $H^q_{1} \triangleq H^q V^q_{1}=U^q_{1} \Sigma^q$. 
Next, the output $y_k$ is decoupled 
using a nonsingular transformation $T^q =\begin{bmatrix} T_{1}^{q\top} & T_{2}^{q\top} \end{bmatrix}^\top \triangleq U^{q\top} =\begin{bmatrix} U^q_{1} & U^q_{2} \end{bmatrix}^\top $ 
to obtain $z^q_{1,k} \in \mathbb{R}^{p_{H^q}}$ and $z^q_{2,k} \in \mathbb{R}^{l-p_{H^q}}$ given by
\begin{gather} \label{eq:sysY} \hspace{-0.2cm}\begin{array}{lll}
z^q_{1,k} &\triangleq T^q_{1} y_k =U_{1}^{q\top} y_k = {C}^q_{1} x_k + \Sigma^q d^q_{1,k} + {D}^q_{k,1} u^q_k + {v}^q_{1,k},\\
z^q_{2,k} &\triangleq T^q_{2} y_k =U_{2}^{q\top} y_k =  {C}^q_{2}  x_k + {D}^q_{k,2} u^q_k + {v}^q_{2,k},
\end{array} 
\end{gather}
where ${C}^q_{1} \triangleq U_1^{q\top} {C}^q$, ${C}^q_{2} \triangleq U_{2}^{q\top} {C}^q$, ${D}^q_{k,1} \triangleq U_{1}^{q\top} {D}^q_k$, ${D}^q_{k,2} \triangleq  U_{2}^{q\top} {D}^q_k$, ${v}^q_{1,k} \triangleq U_{1}^{q\top} {v}^q_k$ and ${v}^q_{2,k} \triangleq  U_{2}^{q\top} {v}^q_k$. This transformation is also chosen such that 
\begin{align*}
\left\|\begin{bmatrix} {{v}^q_{1,k}}^\top & {{v}^q_{2,k}}^\top \end{bmatrix}^\top\right\|_2=\| U^{q\top} {v}^q_k\|\moh{_2}=\|{v}^q_k\|_2. 
\end{align*}
\section{Matrices and Parameters in Proposition \ref{prop:filterbanks}}\label{app:matrices}
\vspace{-1cm}
\begin{align*}
\tilde{Y}^q_1 &\triangleq (P-YC^q_2)\Phi^q, \quad \tilde{Y}^q_2 \triangleq -(P-YC^q_2)\Phi^q\Psi^q, \\ \tilde{\mathbf{M}}_{1} &\triangleq -{\kappa}I-\breve{Q}, \quad \tilde{\mathbf{M}}_{2} \triangleq -{\kappa}{(L_f^q)}^2I+(1-\alpha)P-\tilde{\Gamma}, \quad \tilde{\mathbf{M}}_{3} \triangleq {\kappa}I, \\
\mathcal{N}^q_{21} &\triangleq \Psi^{q\top} \Phi^{q\top} (PR^q-Y\Omega^q-C^{q\top}_2 Y^\top R^q), \\
\mathcal{N}^q_{11} &\triangleq \rho^2 I + 2R^{q\top} Y\Omega^q-R^{q\top} P R^q-\Omega^{q\top} (\Gamma+(\varepsilon^{-1}_1+\varepsilon^{-1}_2)I) \Omega^q, \\
 \mathcal{N}^q_{31} &\triangleq \Phi^{q\top}(Y \Omega^q+C^{q\top}_2 Y^\top R^q -PR^q), \\ 
 \mathcal{N}^q_{33} &\triangleq -\varepsilon_2\Phi^{q\top} C^{q\top}_2 C^q_2 \Phi^q+I ,\\
  \mathcal{N}^q_{22} &\triangleq  -I+\alpha P-\varepsilon_1\Psi^{q\top} \Phi^{q\top} C^{q\top}_2 C^q_2 \Phi^q \Psi^q-{L_f^q}^2I, \\
  {\delta}_\infty^x &\triangleq 
\begin{cases} {{\delta}_{\infty,1}^{x,q}}, & \text{if} \ \theta^q_1 <1,\theta^q_2 \geq 1, \\ {{\delta}_{\infty,2}^{x,q}}, & \text{if} \ \theta^q_1 \geq 1,\theta^q_2 < 1, \\ \min({{\delta}_{\infty,1}^{x,q},{\delta}_{\infty,2}^{q,x}}
), &  \text{if} \ \theta^q_1 < 1,\theta^q_2 < 1, \end{cases}, \\ {\delta}_{\infty,1}^{x,q}& \triangleq \rho^\star_q\sqrt{\frac{{\eta^q_w}^2+{\eta^q_v}^2}{{\lambda_{\min} (P^q)(1-\theta^q_1)}}}, \\
{\delta}_{\infty,2}^{x,q}&\triangleq \frac{ \overline{\eta}^q}{1-\theta^q_2}, \\ 
\theta^q_1 &\triangleq {{\textstyle\frac{|\lambda_{\max}(P^q)-1|}{\lambda_{\min}(P^q)}}}, \\ {\theta_2}^q &\triangleq (L^q_f+\| \Psi^q \|_2)\|(I-\tilde{L}^qC^q_2)\Phi^q \|_2, \\
 \overline{\eta}^q &\triangleq \| \Re^q \|_2 \eta^q_v + \| \Psi^q \Phi^q W^q\|_2 \eta^q_w, \\\Omega^q &\triangleq C^q_2R^q-Q^q,\\ 
 {\delta}^{d,q}_{\infty} &\triangleq \beta^q {\delta}^{x,q}_{\infty}+ \overline{\alpha}^q, \\
  \Re^q &\triangleq -(\Psi^q \Phi^q G^q_1 M^q_1 T^q_1 +\Psi^q G^q_2 M^q_2 T^q_2 + \tilde{L}^q T^q_2), \\
   \beta^q &\triangleq \|V_1^qM^q_1C^q_1-V^q_2M^q_2C^q_2\Psi^q \|_2+L^q_f \|V^q_2M^q_2C^q_2 \|_2, \\
   \overline{\alpha}^q &\triangleq \| V^q_2 M^q_2 C^q_2 \|_2 \eta^q_w
 +\big{[}\|(V^q_2 M^q_2 C^q_2G^q_1-V^q_1)M^q_1T^q_1\|\yong{_2}+\|V^q_2M^q_2T^q_2\|\yong{_2} \big{]} \eta^q_v.
\end{align*}

\end{document}